\documentclass[numberwithinsect,a4paper,UKenglish,cleveref, autoref, thm-restate]{lipics-v2021}

\usepackage{hyperref} 
\usepackage{amsmath, amssymb, amscd, amsthm, amsfonts}
\usepackage{algorithm,algorithmicx,algpseudocode}
\usepackage{tikzsymbols}
\usepackage{pifont}
\usepackage[percent]{overpic}

\usepackage[noadjust,compress]{cite}

\usepackage{thm-restate}

\usepackage{mathtools}

\usepackage{qtree}
\usepackage[linguistics]{forest}

\usepackage{xcolor}

\usepackage{etoolbox}
\usepackage{color,colortbl}
 \usepackage{array}

\usepackage{bbm}
\usepackage{comment}
\usepackage{enumerate}

\usepackage{xspace}
\usepackage{rotating}
\usepackage{xifthen}
\usepackage{url}
\usepackage{csquotes}
\usepackage{graphicx}
\usepackage{wrapfig}
\usepackage{multirow}
\usepackage[binary-units=true]{siunitx}
\usepackage{stmaryrd}

\usepackage{float}
\usepackage[utf8]{inputenc} 
\usepackage{dsfont}

\usepackage[normalem]{ulem}
\usepackage{pgf}
\usepackage{tikz}
\usetikzlibrary{trees,decorations,arrows,automata,shadows,positioning,plotmarks,backgrounds,shapes}
\usetikzlibrary{calc,matrix,fit,petri,decorations.markings,decorations.pathmorphing,patterns,intersections,decorations.text}
\usepackage{pgfplots}
\usepackage{pgfplotstable}

\tikzstyle{mystate}=[state,inner sep=3pt,minimum size=20pt,line width=0.2mm]
\tikzstyle{fstate}=[state,accepting,inner sep=2pt,minimum size=3pt]
\tikzstyle{istate}=[state,initial,inner sep=2pt,minimum size=3pt]
\tikzstyle{mysquare}=[inner sep=3pt,minimum size=15pt,line width=0.2mm]
\tikzstyle{fmysquare}=[inner sep=3pt,minimum size=15pt,line width=0.5mm,accepting]

\usepackage{longtable}
\usepackage{caption}
\usepackage[position=b]{subcaption}

\usepackage{tabularx}
\usepackage{booktabs}
\usepackage{complexity}

\usepackage{etoc}
\etocsettocdepth{3}

\usepackage{xfrac}
\usepackage{faktor}
\usepackage{textcomp}

\usepackage{newfile}
\usepackage{totcount}

\regtotcounter{@todonotes@numberoftodonotes}

\newclass{\logspace}{logspace}

\newcommand{\N}{\mathbb{N}}

\newcommand{\Z}{\mathbb{Z}}

\usepackage[disable,colorinlistoftodos,prependcaption,textsize=tiny]{todonotes}

\newcommand{\ISinside}[1]{\todo[linecolor=green,backgroundcolor=green!35,bordercolor=green]{\textbf{IS:} #1}}

\newcommand{\at}{\alpha}
\newcommand{\atb}{\beta}

\newcommand{\NFA}{\mathcal{A}}
\newcommand{\NFAscc}{\mathcal{R}}
\newcommand{\NFAidl}{\mathcal{R}^{\idl}}
\newcommand{\atoms}[1]{\mathtt{atoms(}#1\mathtt{)}}
\newcommand{\deltaidl}{\delta^\idl}
\newcommand{\Idl}{\mathtt{Idl}}
\newcommand{\idl}{\mathtt{idl}}
\newcommand{\red}{\mathtt{red}}
\newcommand{\NFAred}{\mathcal{R}^{\red}}
\newcommand{\Qred}{Q^\red}
\newcommand{\deltared}{\delta^\red}
\newcommand{\Fred}{F^\red}
\newcommand{\qinitred}{{q_0}^\red}
\newcommand{\qfinalred}{{q_f}^\red}

\newcommand{\maxlength}{m}
\newcommand{\mymatrix}{\mathcal{M}}

\newcommand{\emp}{\varepsilon}

\makeatletter
\newcommand*{\encircled}[1]{\relax\ifmmode\mathpalette\@encircled@math{#1}\else\@encircled{#1}\fi}
\newcommand*{\@encircled@math}[2]{\@encircled{$\m@th#1#2$}}
\newcommand*{\@encircled}[1]{%
  \tikz[baseline,anchor=base]{\node[draw,circle,outer sep=0pt,inner sep=.2ex] {#1};}}
\makeatother

\newcommand{\circqm}{\encircled{?}}
\newcommand{\circstr}{\encircled{*}}

\newcommand{\atomi}{a^{\circqm}}
\newcommand{\atomt}{\Delta^{\circstr}}
\newcommand{\Atomi}[1]{{#1}^{\circqm}}
\newcommand{\Atomt}[1]{{#1}^{\circstr}}

\newcommand{\langatomi}{\{a, \varepsilon\}}
\newcommand{\langatomt}{\Delta^*}
\newcommand{\atom}{\alpha}

\newcommand{\wordatom}[1]{w_{#1}}
\newcommand{\worddelta}[1]{\mathfrak{w}_{#1}}

\newcommand{\pro}{\ensuremath{P}\xspace}
\newcommand{\alp}{\ensuremath{\texttt{alph}}\xspace}

\newcommand{\lA}{\cdot A}
\newcommand{\rA}{ A \cdot}
\newcommand{\rightset}{\texttt{r}}
\newcommand{\leftset}{\texttt{l}}
\newcommand{\Qidl}{Q^\idl}
\newcommand{\qinitidl}{{q_0}^\idl}
\newcommand{\Fidl}{F^\idl}
\newcommand{\cpy}{\mathtt{c}}

\newcommand{\trQ}{Q}
\newcommand{\trAlph}{\ensuremath{\Gamma}\xspace}
\newcommand{\trinit}{\ensuremath{t^0}\xspace}
\newcommand{\trF}{\ensuremath{F}\xspace}
\newcommand{\trE}{E}
\newcommand{\tr}{\mathcal{T}}
\DeclareMathOperator{\val}{val}
\DeclareMathOperator{\inp}{inp}
\DeclareMathOperator{\out}{out}

\newcommand{\trl}{\mathcal{T}_{\mathcal{L}}}
\newcommand{\trr}{\mathcal{T}_{\mathcal{R}}}
\newcommand{\LL}{\mathcal{L}}

\newcommand{\SSI}{\mathbb{S}}

\newcommand{\ltrQ}{\trQ_\LL}
\newcommand{\ltrAlph}{\ensuremath{\Gamma_\LL}\xspace}
\newcommand{\ltrinit}{\ensuremath{\trinit_\LL}\xspace}
\newcommand{\ltrF}{\ensuremath{\trF_\LL}\xspace}
\newcommand{\ltrE}{E_\LL}
\newcommand{\ltr}{\mathcal{T}_\LL}

\newcommand{\twobar}[1]{\bar{\bar{#1}}}

\newcommand{\Gidl}{\ensuremath{G^\idl}\xspace}
\newcommand{\Gred}{\ensuremath{G^{\red}}\xspace}
\newcommand{\Nred}{\ensuremath{N^{\red}}\xspace}
\newcommand{\Sred}{\ensuremath{S^{\red}}\xspace}

\newcommand{\prored}{\ensuremath{\pro^{\red}}\xspace}

\newcommand{\sw}{\preccurlyeq}
\newcommand{\dcl}[1]{#1\mathord{\downarrow}}
\newcommand{\sharpP}{\ComplexityFont{\#P}}
\newcommand{\spanL}{\ComplexityFont{spanL}}
\newcommand{\compP}{\ComplexityFont{P}}

\title{Directed Regular and Context-Free Languages}

\authorrunning{M. Ganardi, I. Sa\u{g}lam, G. Zetzsche} 

\author{Moses Ganardi}{Max Planck Institute for Software Systems (MPI-SWS), Germany}{ganardi@mpi-sws.org}{https://orcid.org/0000-0002-0775-7781}{}

\author{Irmak Sa\u{g}lam}{Max Planck Institute for Software Systems (MPI-SWS), Germany}{isaglam@mpi-sws.org}{https://orcid.org/0000-0002-4757-1631}{}

\author{Georg Zetzsche}{Max Planck Institute for Software Systems (MPI-SWS), Germany}{georg@mpi-sws.org}{https://orcid.org/0000-0002-6421-4388}{}

\Copyright{Moses Ganardi, Irmak Sa\u{g}lam, Georg Zetzsche} 

\ccsdesc[500]{Theory of computation~Problems, reductions and completeness}
\ccsdesc[300]{Theory of computation~Concurrency}
\ccsdesc[500]{Theory of computation~Formal languages and automata theory}

\keywords{Subword, ideal, language, regular, context-free, equivalence, downward closure, compression} 

\nolinenumbers 

\funding{\flag[3cm]{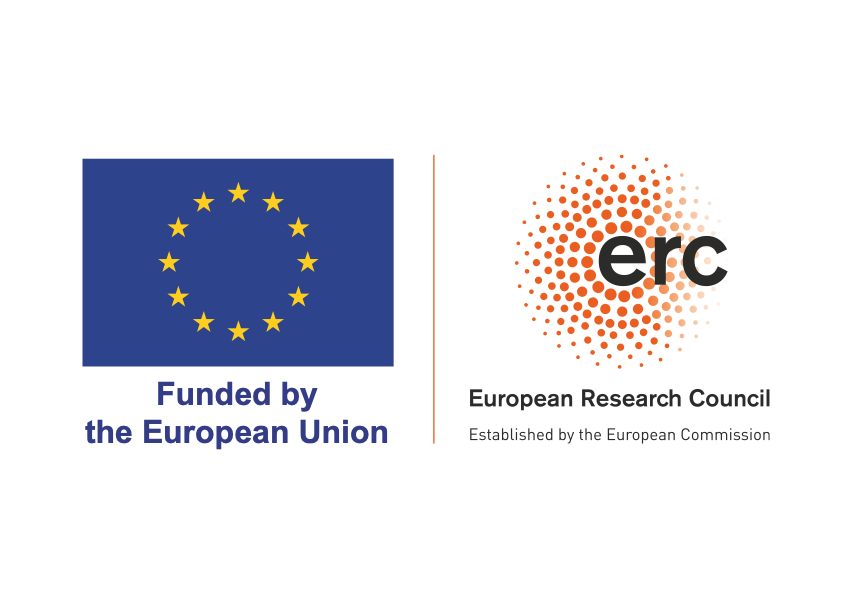}Funded by the European Union (ERC, FINABIS, 101077902). Views and opinions expressed are however those of the author(s) only and do not necessarily reflect those of the European Union or the European Research Council Executive Agency. Neither the European Union nor the granting authority can be held responsible for them.}

\begin{document}
\thispagestyle{empty} 

\maketitle

\begin{abstract}
We study the problem of deciding whether a given language is directed.  A
language $L$ is \emph{directed} if every pair of words in $L$ have a common (scattered)
superword in $L$.  Deciding directedness is a fundamental problem in connection
with ideal decompositions of downward closed sets.
Another motivation is that deciding whether two \emph{directed} context-free
languages have the same downward closures can be decided in polynomial time,
whereas for general context-free languages, this problem is known to be $\coNEXP$-complete.

We show that the directedness problem for regular languages, given as NFAs, belongs to
$\AC^1$, and thus polynomial time.  Moreover, it is $\NL$-complete for fixed
alphabet sizes. Furthermore, we show that for context-free languages,
the directedness problem is $\PSPACE$-complete.
\end{abstract}

\newpage
\clearpage 
\pagenumbering{arabic} 
\section{Introduction}
We study the problem of deciding whether a given language is directed.
A language $L$ is called \emph{(upward) directed} if for every $u,v\in L$, there exists a $w\in L$ with $u\sw w$ and $v\sw w$. Here, $\sw$ denotes the (non-contiguous) \emph{subword} relation: We have $u\sw v$ if there are decompositions $u=u_1\cdots u_n$ and $v=v_0u_1v_1\cdots u_nv_n$ for some $u_1,\ldots,u_n\in\Sigma^*$ and $v_0,v_1,\ldots,v_n\in\Sigma^*$.
\vspace{-3mm}
\subparagraph{Downward closures and ideals}
The \emph{downward closure} of a language $L\subseteq\Sigma^*$ is the set
$\dcl{L}=\{u\in\Sigma^* \mid \exists v\in L\colon u\sw v\}$. Over the last ca.\ 15 years, downward closures have been used in several approaches to verifying concurrent systems. This has two reasons: First, $\dcl{L}$ is a regular language for every set $L\subseteq\Sigma^*$~\cite{Haines69} and an NFA can often be computed effectively~\cite{DBLP:conf/icalp/Zetzsche15,icalpZetzsche16,DBLP:conf/lics/AtigCHKSZ16,DBLP:conf/popl/HagueKO16,DBLP:conf/lics/ClementePSW16,DBLP:conf/icalp/HabermehlMW10,Courcelle91,vanLeeuwen1978}. Second, many verification tasks are \emph{downward closure invariant} w.r.t.\ subsystems: This means, a (potentially infinite-state) subsystem (e.g.\ a recursive program represented by a context-free language) can be replaced with another with the same downward closure, without affecting the verified property. This has been applied to parameterized systems with non-atomic reads and writes~\cite{DBLP:conf/concur/TorreMW15}, concurrent programs with dynamic thread creation~\cite{DBLP:journals/corr/abs-1111-1011,DBLP:conf/icalp/BaumannMTZ20,DBLP:conf/icalp/0001GMTZ23}, asynchronous programs~\cite{DBLP:journals/lmcs/MajumdarTZ22,DBLP:conf/icalp/0001GMTZ23a}, and thread pools~\cite{DBLP:journals/pacmpl/BaumannMTZ22}.

In addition to finite automata, there is a second representation of downward
closed languages: Every non-empty downward closed set can be written as a
finite union of ideals. An \emph{ideal} (in the terminology of well-quasi
orderings) is a non-empty downward closed set that is directed. Moreover,
ideals have a simple representation themselves: They are precisely the products
of languages of the forms $\{a,\varepsilon\}$ and $\Delta^*$, where $a$ is a
letter and $\Delta$ is an alphabet. 
Clearly, a
language $L$ is directed if and only if $\dcl{L}$ is itself an ideal.

Ideal decompositions of downward closed sets have recently been the center of
significant attention: They have been instrumental in computing downward
closures~\cite{DBLP:conf/icalp/Zetzsche15,DBLP:journals/corr/abs-1904-10703,DBLP:conf/icalp/BarozziniCCP20}
and deciding separability by piecewise testable
languages~\cite{DBLP:conf/icalp/Goubault-Larrecq16,DBLP:conf/lics/Zetzsche18}.
Over other orderings, ideals play a crucial role in forward analysis of
well-structured transition systems
(WSTS)~\cite{DBLP:conf/stacs/FinkelG09,DBLP:conf/icalp/FinkelG09,DBLP:conf/fsttcs/BlondinFG17},
infinitely branching WSTS~\cite{DBLP:journals/iandc/BlondinFM18}, well-behaved
transition systems~\cite{DBLP:journals/lmcs/BlondinFM17} for clarifying
reachability problems in vector addition
systems~\cite{DBLP:journals/iandc/LazicS21,DBLP:conf/lics/LerouxS15,DBLP:conf/stacs/LerouxS16},
and for deciding regular
separability~\cite{DBLP:conf/concur/CzerwinskiLMMKS18}.

Given the importance of ideals, it is a fundamental problem to decide whether a
given language is directed, in other words, whether the ideal decomposition of
its downward closure consists of a single ideal. It is a basic task for
computing ideal decompositions, but also an algorithmic lens on the structure
of ideals.

\vspace{-3mm}
\subparagraph{Efficient comparison} Aside from being a fundamental property,
checking directedness is also useful for deciding equivalence.  It is
well-known that equivalence is $\PSPACE$-complete for NFAs and undecidable
for context-free languages.  However, in some situations, it suffices to decide
\emph{downward closure equivalence}: Due to the aforementioned downward closure
invariance in concurrent programs, if $L_1,L_2\subseteq\Sigma^*$ describe the
behaviors of sequential programs inside of a concurrent program, and we have
$\dcl{L_1}=\dcl{L_2}$, then $L_1$ can be replaced with $L_2$ without affecting
safety, boundedness, and termination properties in
concurrent~\cite{DBLP:conf/icalp/0001GMTZ23} and asynchronous
programs~\cite{DBLP:journals/lmcs/MajumdarTZ22}. Downward closure equivalence
is known to be $\coNP$-complete for NFAs~\cite[Thm.~12\&13]{DBLP:conf/lata/BachmeierLS15} and $\coNEXP$-complete for
context-free languages~\cite{icalpZetzsche16}.  This is better than $\PSPACE$
and undecidable, but our results imply that if $L_1$ and $L_2$ are directed, then deciding
$\dcl{L_1}=\dcl{L_2}$ is \emph{in polynomial time}, both for NFAs and
for context-free languages!  
Thus, directedness drastically reduces the complexity of downward closure equivalence.

\vspace{-3mm}
\subparagraph{Constraint satisfaction problems}
Directedness has recently also been studied in the context of constraint satisfaction problems (CSPs) for infinite structures. 
If we view finite words in the usual way as finite relational structures (as in
first-order logic), then a set of words is directed if and only if it has the
\emph{joint embedding property (JEP)}. More generally, a class $\mathcal{C}$ of 
finite structures has the JEP if for any two structures in $\mathcal{C}$, there
is a third in which they both embed.  The JEP is important for CSPs because if
$\mathcal{C}$ is definable by a universal first-order formula and has the JEP,
then it is the \emph{age} of some (potentially infinite) structure, which then
has a constraint satisfaction problem in
$\NP$~\cite[p.~1]{DBLP:journals/dmtcs/BodirskyRS21}.

Motivated by this, it was recently shown that the JEP is undecidable for
universal formulas by Braunfeld~\cite{DBLP:journals/dmtcs/Braunfeld19} (and
even for universal Horn formulas by Bodirsky, Rydval, and
Schrottenloher~\cite{DBLP:journals/dmtcs/BodirskyRS21}). In the special case of
finite words, the JEP (and thus directedness) was shown to be decidable in
polynomial time by Atminas and Lozin~\cite{DBLP:conf/dlt/AtminasL22} for
regular languages of the form $\{w\in\Sigma^* \mid w_1,\ldots,w_n\not\sw w\}$
for given $w_1,\ldots,w_n\in\Sigma^*$. However, to our knowledge, for general
regular languages (or even context-free languages), the complexity is not known.

\vspace{-3mm}
\subparagraph{Contribution}
Our first main result is that for NFAs, directedness is decidable in $\AC^1$, a circuit complexity class within polynomial time, defined by Boolean circuits of polynomial size, logarithmic depth, and unbounded fan-in. If we fix the alphabet size, directedness becomes $\NL$-complete. Our second main result is that for context-free languages, directedness is $\PSPACE$-complete, and hardness already holds for input alphabets of size two.

The proof techniques for the main results also yield algorithms for downward
closure equivalence.  Given $L_1$ and $L_2$, we show that deciding
$\dcl{L_1}=\dcl{L_2}$ is in $\AC^1$ if $L_1$ and $L_2$ are directed and given
by NFAs.  As above, we obtain $\NL$-completeness for fixed alphabets.
If $L_1$ and $L_2$ are context-free languages, then deciding
$\dcl{L_1}=\dcl{L_2}$ becomes $\compP$-complete.  

Finally, we mention that counting the number of ideals in the ideal
decomposition of $\dcl{L}$ is $\sharpP$-complete if $L$ is given as an NFA.
Here, hardness follows from $\sharpP$-hardness of counting words in NFAs of a
given length, and $\sharpP$-membership is a consequence of our methods.

\vspace{-3mm}
\subparagraph{Key ingredients}
The upper bounds as well as the lower bounds in our results rely on new techniques.
With only slight extensions of existing techniques, one would obtain an $\NP$
upper bound for regular languages and an $\NEXP$ upper bound for context-free
languages. This is because given a regular or context-free language, one can
construct an acyclic graph where every path corresponds to an ideal of its
downward closure. If the input is an NFA, this graph is polynomial-sized, and
for CFGs, it is exponential-sized.  One could then guess a path and verify that
the entire language is included in this candidate ideal. 

To obtain our upper bounds, we introduce a \emph{weighting technique}, where
each ideal is assigned a weight in the natural numbers. The weighting function
has the property that if there is an ideal that contains the entire language,
it must be one with maximal weight. Using either (i)~matrix powering  over the
semiring $(\N\cup\{-\infty\},\max,+,-\infty,0)$ for NFAs or (ii)~dynamic
programming for context-free grammars, this allows us to compute in
$\NL$/$\AC^1$/polynomial time a unique candidate ideal (which is compressed in
the context-free case), which is then verified in $\NL$ resp. in $\PSPACE$.

For the $\PSPACE$ lower bound for context-free languages, we first observe that
directedness is equivalent to deciding whether $L\subseteq I$, where $L$ is a
context-free language, and $I$ is an SLP-compressed ideal.  The problem is thus
a slight generalization of the \emph{compressed subword problem}, where we are
given two SLP-compressed words $u$ and $v$ and are asked whether $u\sw v$. This
problem is known to be in $\PSPACE$ and
$\PP$-hard~\cite[Theorem~13]{DBLP:journals/iandc/Lohrey11}, but its exact
complexity is a long-standing open problem~\cite{Lohrey12}. 

To exploit the increased generality of our problem, we proceed as follows.  Our
key insight is that for a given SLP-compressed word $w\in\Sigma^*$, one can
construct an SLP-compressed infinite \emph{complement ideal} $I_w$, meaning
that $I_w\cap\Sigma^{|w|}=\Sigma^{|w|}\setminus\{w\}$.  To this end, we apply a
construction from definability of languages in the subword
ordering~\cite[Lemma~3.1]{DBLP:conf/stacs/BaumannGTZ22}. We use the complement
ideal for $\PSPACE$-hardness follows.  We reduce from the $\PSPACE$-complete
problem of deciding, given two equal-length SLP-compressed words
$v,w\in\{a,b\}^*$, whether their convolution $v\otimes w$ belongs to a fixed
regular language~\cite[p.~269]{Lohrey12}.  We reduce this to the problem of
deciding $w\in L$, where $w$ is SLP-compressed and $L$ is a context-free of
words of length $|w|$.  Then $L\subseteq I_w$ if and only if $w\notin L$, hence
$L\cup I_w$ is directed if and only if $w\notin L$.


\vspace{-2mm} 
\section{Main results}
\vspace{-2mm}
Our first main result is that for a given NFA, one can decide directedness of
its language in polynomial time, and even in $\AC^1 \subseteq \NC$.
\begin{theorem}\label{thm:main-regularAC1}
Given an NFA, one can decide in $\AC^1$ whether its language is directed.
\end{theorem}
Recall that $\AC^1$ is the class of all languages that are accepted by
a family of unbounded fan-in Boolean circuits of polynomial size and logarithmic depth,
see~\cite{DBLP:books/daglib/0097931} for more details.
In particular, the directedness problem for regular languages can be efficiently parallelized.

The same techniques show that fixing the input alphabet leads to $\NL$-completeness:
\begin{theorem}\label{thm:main-regularNL-complete}
For every fixed $k$, given an NFA over $k$ letters, it is $\NL$-complete to
decide whether its language is directed.
\end{theorem}
As mentioned before, slight extensions of known techniques would yield an $\NP$
upper bound for directedness of regular languages: Given an NFA $\NFA$, it is
not difficult construct an acyclic graph whose paths correspond to ideals
for which $L(\NFA)=I_1\cup\cdots\cup I_n$. One can then guess
such a path with ideal $I_i$ and verify $L(\NFA)\subseteq I_i$ in $\NL$.
Clearly, $L(\NFA)$ is directed iff such an $I_i$ exists. The key
challenge is to compute in $\AC^1$ a single ideal $I_i$ for which we
check $L(\NFA)\subseteq I_i$.

Note that \cref{thm:main-regularAC1,thm:main-regularNL-complete}
also apply to one-counter languages. Given a one-counter language $L$,
one can compute in $\logspace$ an NFA $\NFA$ with $L(\NFA)=\dcl{L}$~\cite[Theorem
7]{DBLP:conf/lics/AtigCHKSZ16}\footnote{Theorem 7 in
	\cite{DBLP:conf/lics/AtigCHKSZ16} only states a polynomial time
computation, but it is clear that it can be performed in (deterministic)
logspace.}. Then $L(\NFA)$ is directed iff $L$ is directed, and we can just
decide directedness for $\NFA$.

Our second main result is that for context-free languages (given, e.g.\ by a
grammar), directedness is $\PSPACE$-complete:
\begin{theorem}\label{thm:main-contextfree}
Given a context-free grammar, it is $\PSPACE$-complete to decide whether its
language is directed. Moreover, $\PSPACE$-hardness holds already for binary
input alphabets.
\end{theorem}

Our methods also provide more efficient algorithms for downward closure
equivalence in the case of directed input languages. 
The \emph{downward closure equivalence (DCE)} problem is to decide, for given languages $L_1$ and $L_2$, whether $\dcl{L_1}=\dcl{L_2}$.
While DCE is $\coNP$-complete for general regular languages given as
NFAs~\cite[Thm.~12\&13]{DBLP:conf/lata/BachmeierLS15}, we show that for
directed input languages, the complexity drops to $\AC^1$, and $\NL$ for fixed
alphabets.
\begin{theorem}\label{directed-nfa-dce}
	For directed languages given as NFAs, DCE belongs to $\AC^1$, for fixed alphabets even to $\NL$.
\end{theorem}

For context-free languages, DCE is known to be
$\coNEXP$-complete~\cite{icalpZetzsche16}. For directed input languages, our
methods yield a drastic drop in complexity down to polynomial time: 
\begin{theorem}\label{directed-cfl-dce}
	For directed context-free languages, DCE is $\compP$-complete.
\end{theorem}

Since our directedness algorithms decide whether the unique decomposition of
$\dcl{L}$ into maximal ideals consists of a single ideal, it is natural to ask
about the complexity of counting all ideals of this decomposition. It follows
easily using methods developed here that this problem is in $\sharpP$.
Moreover, the well-known $\sharpP$-hardness of \#NFA (i.e.~counting words of a given length
in an NFA)~\cite{DBLP:journals/tcs/AlvarezJ93} provides a $\sharpP$ lower bound.
\begin{restatable}{theorem}{restatablecountingideals}\label{thm:counting-ideals}
	Given an NFA $\NFA$, it is $\sharpP$-complete to count the number of
	ideals in the decomposition of $\dcl{L(\NFA)}$ into maximal ideals.
\end{restatable}
The proof is given in~\cref{sec:appendix-counting-ideals}.
 $\sharpP$
is the class of functions $f$ computable by some
non-deterministic polynomial-time Turing machine (TM), in the sense that for a given
input $x$, the computed value $f(x)$ is the number of accepting runs.
$\sharpP$-complete problems are very hard, as evidenced by Toda's well-known
result that $\compP^{\sharpP}$ (i.e.~polynomial-time algorithms with access to
$\sharpP$ oracles) includes the entire polynomial time hierarchy~\cite{TodaPP}.
This represents an interesting contrast between the complexity of directedness
and counting all ideals.

\vspace{-2mm}

\section{Preliminaries}\label{sec:preliminaries}

\vspace{-3mm}
\subparagraph{Ideals}
We will use the notation $[m_1, m_2] := \{i \in \Z \mid m_1 \leq i \leq m_2 \}$.
Consider the context-free languages $K_1=\{wcw \mid w\in\{ab\}^*\}$ and $K_2=\{wcw \mid w\in\{a\}^*\cup\{b\}^*\}$. Note that $K_1$ is directed, whereas $K_2$ is not: The words $aca$ and $bcb$ in $K_2$ have no common superword in $K_2$. However, $K_1\cup K_2$ is directed.
Let $\Sigma$ be a finite alphabet. 
A set $D\subseteq\Sigma^*$ is \emph{downward closed} if $\dcl{D}=D$. 
A subset $I\subseteq\Sigma^*$ is an \emph{ideal} if it is non-empty, downward closed, and (upward) directed. Thus clearly, a non-empty $L\subseteq\Sigma^*$ is directed if and only if $\dcl{L}$ is an ideal (note that taking the downward closure does not affect directedness).
It is known that every ideal can be written as products of so called \emph{atoms}:
We identify two types of atoms over $\Sigma$:
\begin{description}
\item[Single atoms:] $\atomi$ where $ a\in \Sigma$,
\item[Alphabet atoms:] $\atomt$ where $\emptyset \neq \Delta \subseteq \Sigma$.
\end{description}
Formally, each atom $\at$ is a formal symbol that describes an ideal $\Idl(\at)$. For a single atom $\atomi$, we define it as $\Idl(\atomi) = \langatomi$, whereas for an alphabet atom $\atomt$, $\Idl(\atomt) = \langatomt$. By $\atoms{\Sigma}$, we denote the set of atoms over $\Sigma$. Note that $|\atoms{\Sigma}|=2^{|\Sigma|}-1+|\Sigma|$. 

An \emph{ideal representation} is a finite (possibly empty) sequence $r = \at_1 \cdots \at_n$ of atoms $\at_i$.
Its language is the concatenation $\Idl(\at_1 \cdots \at_n) = \Idl(\at_1) \cdots \Idl(\at_n)$
where the empty concatenation is interpreted as $\{\varepsilon\}$.
It is a classical fact that every downward closed set can be decomposed into a finite union of ideals. 
For example, observe that $\dcl{K_1}=\dcl{(K_1\cup K_2)}=\{a,b\}^*\{c,\varepsilon\}\{a,b\}^*=\Idl(\Atomt{\{a,b\}}\Atomi{c}\Atomt{\{a,b\}})$ and $\dcl{K_2}=\{a\}^*\{c,\varepsilon\}\{a\}^*\cup \{b\}^*\{c,\varepsilon\}\{b\}^*=\Idl(\Atomt{\{a\}}\Atomi{c}\Atomt{\{a\}})\cup\Idl(\Atomt{\{b\}}\Atomi{c}\Atomt{\{b\}})$.
This decomposition result was first shown by Jullien~\cite{Jullien1968} and the equivalent fact that every downward closed set can be expressed as a \emph{simple regular expression} was shown independently by Abdulla, Bouajjani, and Jonsson~\cite{DBLP:conf/cav/AbdullaBJ98} (see~\cite{thesisHalfon} for a general treatment).
If $R$ is a set of ideal representations, we set $\Idl(R) := \bigcup_{r \in R} \Idl(r)$.

\vspace{-2mm}
\subparagraph{Reduced ideal representations}
Note that one ideal can have multiple different representations.
For instance, the representations $\Atomt{a}\cdot\Atomi{a}\cdot\Atomi{b}\cdot\Atomt{b}\cdot\Atomi{a}$ and $\Atomt{a}\cdot\Atomt{b}\cdot\Atomi{a}$ \emph{represent} the same ideal,
namely all words that start with a (possibly empty) sequence of $a$'s, followed by a (possibly empty) sequence of $b$'s,
and possibly end with an $a$.
This is because in the first representation, all the words $\Atomi{a}$ and $\Atomi{b}$ generate are produced by their neighboring alphabet atoms.
Two representations are called \emph{equivalent} if they represent the same ideal. 

To achieve unique ideal representations, one can use \emph{reduced
representations}, which we define next.  Two atoms $\at$ and $\atb$ are
\emph{absorptive} if $\Idl(\at \atb) = \Idl(\at)$ or $\Idl(\at \atb) = \Idl(\atb)$. In the first
case we say \emph{$\at$ absorbs $\atb$} and in the second case, \emph{$\atb$ absorbs
$\at$}.  Note that two single atoms are always non-absorptive since
$\Idl(\Atomi{a}\cdot\Atomi{b}) \supsetneq \Idl(\Atomi{a})$. An atom $\at$ is
said to \emph{contain} an atom $\atb$ if $\Idl(\at) \supseteq \Idl(\atb)$. $\alpha$ is said
to \emph{strictly} contain $\atb$ if also $\Idl(\at) \neq \Idl(\atb)$.  An ideal
representation $\at_1\cdots \at_n$ is said to be \emph{reduced} if for all $i \in
[1, n-1]$, $\at_i$ and $\at_{i+1}$ are non-absorptive. The following is obvious (and well-known~\cite[Lemma~5.4]{DBLP:journals/fmsd/AbdullaCBJ04}),
because we can just repeatedly merge neighboring absorptive atom pairs:
\begin{restatable}{lemma}{restatableexistsnormalform}\label{thm:existsnormalform}
    For every ideal representation $\at_1 \cdots \at_n$, there exists a reduced ideal representation $\atb_1 \cdots \atb_m$ such that $\Idl(\at_1 \cdots \at_n) = \Idl(\atb_1 \cdots \atb_m)$ and $m \leq n$. 
\end{restatable}


 \vspace{-5mm}
 
\subparagraph{Representing downward closed sets}
We will use two classical facts about ideals. First, every downward closed set $D \subseteq \Sigma^*$ can be written as a finite union of ideals. Moreover, ideals are ``prime'' in the sense that if an ideal is included in a union $D_1\cup D_2$ of downward closed sets, it is already included in one of them:
\begin{lemma}[\cite{DBLP:conf/stacs/FinkelG09,DBLP:journals/ita/KabilP92,DBLP:journals/corr/abs-1904-10703}]\label{ideals-basic-facts}
For every downward closed set $D\subseteq\Sigma^*$, there exist $n\in\N$ and ideals $I_1,\ldots,I_n\subseteq\Sigma^*$ with $D=I_1\cup\cdots\cup I_n$. Moreover, if $I$ is an ideal with $I\subseteq D_1\cup D_2$ for downward closed $D_1,D_2\subseteq\Sigma^*$, then $I\subseteq D_1$ or $I\subseteq D_2$.
\end{lemma}
The representation $D=I_1\cup\cdots\cup I_n$ is also called an \emph{ideal decomposition} of $D$. Observe that the second statement implies that this decomposition is unique (up to the order of ideals) if we require the ideals $I_1,\ldots,I_n$ to be pairwise incomparable. This is sometimes called the unique \emph{decomposition into maximal ideals}.

\vspace{-2mm}

\subparagraph{Non-deterministic finite automata} 
We start by formally introducing NFAs.
A \emph{non-deterministic finite automaton (NFA)} is a tuple 
$\NFA=(Q, \Sigma, \delta, q_0, F)$ where $Q$ is a finite set of \emph{states}, $q_0 \in Q$ is the unique \emph{initial state}, $F \subseteq Q$ is the set of \emph{final states}, $\Sigma$ is a \emph{finite alphabet} and $\delta \subseteq Q \times (\Sigma \cup \{\emp\}) \times Q$ is the set of \emph{transitions}. A transition $(p,a,q) \in \delta$ is usually displayed as $p \xrightarrow{a} q$,
and we write $p_0 \xrightarrow{w} p_n$ if there exists a sequence of transitions $p_0 \xrightarrow{a_1} p_1 \xrightarrow{a_2} \dots \xrightarrow{a_n} p_n$ such that $w = a_1 \dots a_n$.
The language accepted by an NFA $\NFA$ is the set of all words $w \in \Sigma^*$ such that $q_0 \xrightarrow{w} q$ for some $q \in F$, and is denoted by $L(\NFA)$.



\section{Solution on Regular Languages}\label{sec:solutionregular}
In this section, we prove~\cref{thm:main-regularAC1} and~\cref{thm:main-regularNL-complete}.
Let us quickly observe the $\NL$ lower bound of the directedness problem: We reduce from the emptiness problem for NFAs. Given an NFA $\NFA$, we may assume that there is exactly one final state that is different from the initial state, and that all edges are labeled with $a$. We construct an NFA $\NFA'$ with $L(\NFA')=L(\NFA)\cup\{b\}$. Then clearly, $L(\NFA)\ne\emptyset$ if and only if $L(\NFA)$ contains some word in $\{a\}^+$. The latter is true if and only if $L(\NFA')$ is not directed.

Thus, the interesting part of \cref{thm:main-regularAC1,thm:main-regularNL-complete} are the upper bounds.  To explain the main steps, we need some terminology.
    We say that a function $f \colon \{0,1\}^* \to \{0,1\}^*$ is \emph{computable in $\NL$}
    if there exists a non-deterministic $\logspace$ TM with a write-only output tape
    such that (i)~on every input word $x \in \{0,1\}^*$ there exists an accepting computation,
    and (ii)~every accepting computation on $x$ produces the same output $f(x)$ on the output tape.
    We say that $f$ is \emph{computable in $\AC^1$} if the language $\{ x 0 1^i \mid \text{the $i$-th bit of $f(x)$ is $1$}\}$
    belongs to $\AC^1$.
The main difficulty in proving \cref{thm:main-regularAC1} is the following:
\begin{lemma}\label{nfa-ideal-candidate}
Given a non-empty NFA $\NFA$, one can compute in $\AC^1$ an ideal $I$ such that
(i)~$I\subseteq \dcl{L(\NFA)}$ and
(ii)~$L(\NFA)$ is directed if and only if $\dcl{L(\NFA)}\subseteq I$.
Moreover, for every fixed alphabet size, this computation can be carried out in $\NL$.
\end{lemma}
Since the inclusion $\dcl{L(\NFA)}\subseteq I$ for a given NFA $\NFA$ and ideal $I$ can be decided in $\NL$~\cite{icalpZetzsche16}, \cref{nfa-ideal-candidate} immediately implies the upper bounds: Just compute $I$ and check $\dcl{L(\NFA)}\subseteq I$.

Let us briefly outline the proof of \cref{nfa-ideal-candidate}. 
Given an NFA $\NFA$ for a regular language $L$, we first construct an NFA $\NFAidl$ accepting representations ideals in an ideal decomposition of $\dcl{L}$. This is then transformed into an NFA $\NFAred$ that accepts \emph{reduced} ideal representations. 
Reducedness of the ideal representations enables us to efficiently compute a maximal ideal in $L(\NFAred)$, by solving a maximum weight path problem.
This step can be carried out in $\AC^1$ for arbitrary alphabets and in $\NL$ for fixed alphabets.
\cref{fig:running-ex} depicts the mentioned transitionary automata and serves as a running example throughout the section.

\vspace{-2mm}

\subsection{Computing the ideal automaton $\NFAidl$}
Our first step towards \cref{nfa-ideal-candidate} is to transform the input NFA into one that is partially ordered. Here, an NFA is  \emph{partially ordered} if the state set $Q$ is equipped with a partial order $(Q,\le)$ such that for every transition from $p$ to $q$, we have $p\le q$.
In particular, the automaton does not contain any cycles except for self-loops. The following is a standard fact:
\begin{restatable}{lemma}{restatabledownwardclosureNFA}\label{lemma:downwardclosureNFA}
	Given any NFA $\NFA$, one can compute in $\NL$ a partially ordered NFA $\NFAscc$ such that $L(\NFAscc)=\dcl{L(\NFA)}$.
\end{restatable}
Essentially, one collapses each strongly connected component (SCC) $C$ of $\NFA$ into
a new state $q_C$ of $\NFAscc$ and adds a self-loop to $q_C$ for each letter
that appears in $C$. Here, we require non-determinism in our $\logspace$
computation, because we need to determine whether a given letter appears in a
strongly connected component. See~\cref{app:sec:proofsofsecsolutionregular} for details.

Next, we want to construct an NFA $\NFAidl$ over a finite alphabet of \emph{atoms of $\Sigma$}, that will accept (as its \emph{words}) the ideal representations given by the accepted paths of $\NFAscc$. 
\begin{restatable}{lemma}{restatableidealNFA}\label{lemma:idealNFA}
	Given any partially ordered NFA $\NFAscc$ on the finite alphabet $\Sigma$,
	one can compute in $\NL$ an acyclic NFA $\NFAidl$ over some polynomial-sized alphabet $\Gamma\subseteq\atoms{\Sigma}$ such that
	$\Idl(L(\NFAidl)) =  \dcl{L(\NFAscc)}$.
\end{restatable}

\noindent Since the only cycles $\NFAscc$ contains are self loops, we can write $L(\NFAscc)$ as the finite union, 
\begin{equation*}
    L(\NFAscc) = \bigcup_{i \in [1,r]} a_{1,i} \Delta_{1,i}^* a_{2,i} \Delta_{2,i}^* \cdots a_{k_i,i} \Delta_{k_i,i}^* \, \text{ with } a_{n,i} \in \Sigma \cup \{\emp\}\, \text{ and }\,\Delta_{n,i}\subseteq \Sigma \text{ for } n \in [1, k_i]
\end{equation*}

\noindent Since $L(\NFAscc)$ is downward closed, it is equivalent to the following ideal decomposition of $\dcl{L(\NFA)}$:

\vspace{-3mm}
\begin{equation}\label{eq:acc-paths-ideals}
    L(\NFAscc) = \bigcup_{i \in [1, r]} \{\emp, a_{1,i}\} \Delta_{1,i}^* \{\emp, a_{2,i}\} \Delta_{2,i}^* \cdots \{\emp,a_{k_i,i}\} \Delta_{k_i,i}^*
\end{equation}

\begin{proof}[Proof sketch]
To construct $\NFAidl$ from $\NFAscc$, for each state $q$ in $\NFAscc$, we add two copies $q$ and $q'$ to $\NFAidl$. We keep the initial state the same, and make the final states of $\NFAidl$ the copies of final states of $\NFAscc$.
Each state $q$ with self-loops is turned into a transition reading an alphabet atom $q \xrightarrow{\Atomt{\Delta}} q'$
where $\Delta$ contains all letters read on self-loops on $q$.
Furthermore, each transition $p \xrightarrow{a} q$ in $\NFAscc$ where $p \neq q$ is turned into a transition reading a single atom $p' \xrightarrow{\Atomi{a}} q$.

It is easily verified $L(\NFAidl)$ is the ideal decomposition of $L(\NFAscc)$ given in equation~\eqref{eq:acc-paths-ideals}. 
\end{proof}

\subsection{Weighting functions for ideals}
\Cref{lemma:idealNFA} tells us that for our given NFA $\NFA$, we can construct an NFA $\NFAidl$ over $\atoms{\Sigma}$
that is acyclic and whose paths correspond to the ideals of $\dcl{L(\NFA)}$. Observe that $L(\NFA)$ is directed iff there exists a path $\pi$ in $\NFAidl$ such that $L(\NFA)$ is included in the ideal of $\pi$:
Clearly, if there is such a path, then $\dcl{L(\NFA)}$ must equal this ideal and is thus directed.
 Conversely, if $L(\NFA)$ is directed and $\dcl{L(\NFA)}=I_1\cup\cdots\cup I_n$, where $I_1,\ldots,I_n$ are the ideals of the paths in $\NFAidl$, then directedness implies that $\dcl{L(\NFA)}$ is an ideal. By~\cref{ideals-basic-facts}, we must have that $\dcl{L(\NFA)}$ coincides with some $I_i$ for $i \in [1, n]$, which lies on some path $\pi$.
Therefore, to show \cref{nfa-ideal-candidate}, it remains to pick a path $\pi$ such that if there is a greatest ideal among the paths in $\NFAidl$, then it must lie on $\pi$. This is the main challenge in our decision procedures.

Our key insight is that this can be accomplished by a \emph{weighting
function}. Roughly speaking, we construct a function $\mu$ from the set of
ideal representations to $\N$ such that $\mu$ is \emph{strictly monotone},
meaning (i)~if $\Idl(\at_1\cdots \at_n)\subseteq \Idl(\atb_1\cdots \atb_m)$, then $\mu(\at_1\cdots
\at_n)\le \mu(\atb_1\cdots \atb_m)$ and (ii)~if in addition $\Idl(\at_1\cdots \at_n)\ne \Idl(\atb_1\cdots
\atb_m)$, then $\mu(\at_1\cdots \at_n)<\mu(\atb_1\cdots \atb_m)$. Moreover, the function will
be \emph{additive}, meaning $\mu(\at_1\cdots \at_n)=\mu(\at_1)+\cdots+\mu(\at_n)$.
Given such a function, we can find the aforementioned path $\pi$
by picking one with maximal weight.

\vspace{-2mm}
\subparagraph{The weight function}
Strictly speaking, an additive strictly monotone function as above is impossible: It would imply $n\le\mu(\Atomi{a}\cdots \Atomi{a})<\mu(\Atomt{\{a\}})$ for every $n$ (here, the product $\Atomi{a}\cdots\Atomi{a}$ has $n$ factors). Therefore, we will only satisfy strict monotonicity on ideal representations of some maximal length $k$. 
Given $k\in\N$, the \emph{($k$-)weight} of an
ideal representation $\at_1 \cdots \at_n$ is defined as
 $\mu_k(\at_1 \cdots \at_n) = \sum^{n}_{i=1} \mu_k(\at_i)$ where for each atom $\at$:
\begin{align}\label{eq:idealweight}
     \mu_k(\alpha) = \begin{cases}
        1, &\text{ if } \,\, \at \text{ is a single atom }\\
        (k+1)^{|\Delta|}, &\text{ if } \,\, \at = \atomt \text{ for some }\Delta \subseteq \Sigma \text{ where } \Delta \neq \emptyset.     \end{cases}
\end{align}
The function $\mu_k$ is clearly additive.  
However, it is not strictly monotone: For instance, the products $\Atomt{a}\cdot\Atomt{a}$ and $\Atomt{a}\cdot\Atomt{b}$ receive the same weight, but the former represents a strict subset of the latter.
In the remainder of this subsection, we will show that $\mu_k$ is strictly monotone on reduced ideal representations. 

\vspace{-2mm}
\subparagraph{Exponential weights are needed} Before we continue with our algorithm for directedness, let us quickly remark that $\mu_k$ cannot be chosen much smaller. If the alphabet $\Sigma$ is not fixed, then $\mu_k$ can have exponential values, because $|\Delta|$ appears in the exponent. In fact, dealing with exponentially large numbers is the reason our upper bound in the general NFA case is $\AC^1$ rather than $\NL$. This raises the question of whether there is a weighting function that is strictly monotone on reduced ideal representations of length $k$ with polynomial values. This is not the case.
To see this, consider the exponential-length chain $C_\ell$ of ideals over the alphabet $\Sigma = \{a_0, a_1, \ldots, a_\ell\}$ constructed as follows.
For $i \in [1, \ell]$, let $\Sigma_i=\{a_0, \ldots, a_i\}$. Set $C_0 := (\Atomi{a_0})$.
 For each $i \in [1, \ell]$, assuming $C_{i-1} = (I_1, \ldots, I_t)$, define $C_i$ as
 \[ C_{i} := (I_1, \ldots, I_t, \Atomt{(\Sigma_{i-1})}\cdot\Atomi{a_i}\cdot I_1,  \ldots,  \Atomt{(\Sigma_{i-1})}\cdot \Atomi{a_i} \cdot I_t). \]
Clearly, the number of ideals in each chain $C_\ell$ is exponential in $\ell$. Moreover, the maximal length $k$ of ideals in $C_\ell$ is polynomial in $\ell$. Since for each $i$, $a_i$ does not appear in $C_{i-1}$, we know that (a)~the ideal representations in $C_\ell$ are reduced and (b)~the chain $C_\ell$ is strict.
Therefore, any weight function that is strictly monotone on ideal representations of length $k$ maps each ideal in $C_\ell$ to a distinct value, requiring exponentially high values.

\vspace{-2mm}
\subparagraph{Strict monotonicity}
We now prove our strict monotonicity property for $\mu_k$:
\begin{proposition}\label{prop:weightfunc}
Let $\at_1\cdots \at_n$ and $\atb_1 \cdots \atb_m$ be reduced ideal representations of ideals $I$ and $J$, respectively, with $m,n\le k$. If $I \subseteq J$ then $\mu_k(\at_1 \cdots \at_n) \leq \mu_k(\atb_1 \cdots \atb_m)$. 
Moreover, if $I \subsetneq J$, then $\mu_k(\at_1 \cdots \at_n) < \mu_k(\atb_1 \cdots \atb_m)$. 
\end{proposition}

To prove \cref{prop:weightfunc}, we use \cref{lemma:embedding},
which roughly states that inclusion of ideals behaves similarly to the
subword ordering: Inclusion is witnessed by some embedding map.
\begin{restatable}{lemma}{restatableembedding}\label{lemma:embedding}
    Let $\at_1\cdots \at_n$ and $\atb_1 \cdots \atb_m$ be representations of ideals $I$ and $J$ on $\Sigma$, respectively. If $I \subseteq J$, then there exists a function $f \colon [1, n] \to [1, m]$ such that 
    \begin{enumerate}
        \item $f(i)\leq f(i+1)$ for all $i\in [1, n-1]$,  \label{line:one}
        \item $\at_i$ is contained in $\atb_{f(i)}$ for all $i \in [1,n]$, \label{line:two}
        \item if $\atb_j$ is a single atom, then $|f^{-1}(j)|\leq 1$ \label{line:three}
    \end{enumerate}
\end{restatable}
\begin{proof}[Proof sketch] For each atom $\at$, we generate a unique word $\wordatom{\at}$. If $\at =\atomi$, then $\wordatom{\at} = a$. 
If $\at = \atomt$, then we fix an order on $\Sigma$ and for each $\Delta \subseteq \Sigma$ let $\worddelta{\Delta}$ be a word that contains each letter in $\Delta$ once, in the increasing order and set $\wordatom{\at} = \worddelta{\Delta}^{m+1}$. 
We define $f$ so that it sends each $i$ to the $j$ for which $\atb_1 \cdots \atb_j$ is the shortest prefix for which $\wordatom{\at_1} \cdots \wordatom{\at_i} \in \Idl(\atb_1 \cdots \atb_j)$. Then $f$ satisfies the premises of~\cref{line:one}-\ref{line:three}. Details can be found in~\cref{app:sec:proofsofsecsolutionregular}.
\end{proof}

\begin{proof}[Proof of \cref{prop:weightfunc}]
 For the given ideal representations $\at_1 \cdots \at_n$ and $\atb_1 \cdots \atb_m$, let $f\colon [1,n]\to [1,m]$ be the embedding function introduced in~\cref{lemma:embedding}.
    \begin{claim*}
    For all $j \in [1,m]$, $\quad\mu_k(\atb_j) \geq \sum_{i\in f^{-1}(j)} \mu_k(\at_i)$
    \end{claim*}
    \begin{claimproof}[Proof of claim]
        If $\atb_j$ is a single atom, then by~\cref{line:three}, there exists at most one $\at_i$ embedded in $\atb_j$ and by~\cref{line:two}, $\atb_j$ contains $\at_i$; thus $\at_i = \atb_j$ is a single atom. 
        In this case, $\mu_k(\atb_j) = \mu_k(\at_i) = 1$.
        Otherwise, $\atb_j = \atomt$. By~\cref{line:one}, the elements of $f^{-1}(j)$ have to be consecutive numbers, i.e. $f^{-1}(j) = [i_1, i_2]$. Then $f$ embeds $\at_{i_1}, \ldots, \at_{i_2}$ in $\atb_j$. Since $\at_1\cdots \at_n$ is reduced, each pair $\at_i$, $\at_{i+1}$ is non-absorptive. Since they are all contained in $\atb_j$, either $|f^{-1}(j)|=1$, or for all $i \in [i_1, i_2]$, $|\at_i| < |\atb_j|$ where $|\at_i| = 0$ if $\at_i$ is a single atom; otherwise it is the size of the alphabet of $\at_i$. 
        In the case $|f^{-1}(j)| = 1$, since the only atom in $f^{-1}(j)$ is contained in $\atb_j$, the claim trivially holds. In the latter case, the inequality~\eqref{eq:weightsupset} 
        \begin{equation}\label{eq:weightsupset}
            \sum_{i\in f^{-1}(j)} \mu_k(\at_i) \leq n\cdot {(k+1)}^{|\Delta|-1} < (k+1)^{|\Delta|} = \mu_k(\atb_j)
        \end{equation} 
        follows from the fact that $f$ can embed at most $n$ many atoms into $\atb_j$ and that $n \leq k$.   
    \end{claimproof}

    $\mu_k(\atb_1\cdots \atb_m) \geq \mu_k(\at_1\cdots \at_n)$ follows from the claim due to the weight of an ideal representation being defined additively. 
     This concludes the first part of the proof.

    Assume $I \subsetneq J$.  Then there exists an $\at_i$ strictly contained in $\atb_{f(i)}$,
    or $f$ is not surjective. In the latter case, equation~\eqref{eq:strictineq} follows from our previous argument.  
    \begin{equation}\label{eq:strictineq}
    \mu_k(\at_1\cdots \at_n) < \mu_k(\atb_1 \cdots \atb_m)
    \end{equation}
    In the former case, $\atb_{f(i)}$ is an alphabet atom, otherwise it cannot strictly contain $\at_i$. 
    If $\at_i$ is a single atom, \eqref{eq:strictineq} follows from $\mu_k(\at_i)=1$. 
    If it is an alphabet atom, \eqref{eq:strictineq} follows from~\eqref{eq:weightsupset}.
\end{proof}

\subsection{Reducing ideals}\label{subsec:reducingideals}
We will apply the weighting function with $k$ being an upper bound on the path length in $\NFAidl$ (e.g.\ the number of states). We have seen that the weighting function is strictly monotone on \emph{reduced} ideal representations. 
Therefore, if all paths in $\NFAidl$ had reduced ideals, we could prove \cref{nfa-ideal-candidate} by picking the path with the largest weight. This is because, if $I_1,\ldots,I_n$ are the ideals on paths of $\NFAidl$ and $I_m$ has maximal $\mu_k$ among them, then \cref{prop:weightfunc} implies that $L(\NFA)$ is directed if and only if $L(\NFA)\subseteq I_m$: Here, the ``if'' is trivial. Conversely, if $L(\NFA)$ is directed, then $\dcl{L(\NFA)}=I_1\cup\cdots\cup I_n$ is an ideal and hence $I_1\cup\cdots\cup I_n=I_i$ for some $i$ by \cref{ideals-basic-facts}. But then we must have $I_i=I_m$, because $I_m\subseteq I_i$ by the choice of $I_i$, and if $I_i$ were a strict superset of $I_m$, $\mu_k(I_m)$ would not be maximal. Thus, $L(\NFA)\subseteq I_m$.

Thus, our next task is to transform $\NFAidl$ so as to make all ideal representations reduced:
\begin{lemma}\label{lemma:reducedidealNFA}
	Given any partially ordered NFA $\NFAscc$ on the finite alphabet $\Sigma$,
	one can compute in $\NL$ an acyclic NFA $\NFAred$ over some polynomial-sized alphabet $\Gamma \subseteq \atoms{\Sigma}$ such that 
    $\Idl(L(\NFAred)) = \dcl{L(\NFAscc)}$
and the ideal representations $\NFAred$ accepts are reduced.
\end{lemma}
%
%
%
%
Reducing an individual ideal representation is easy: repeatedly merge consecutive atoms, as briefly sketched in~\cref{thm:existsnormalform}. 
Reducing \emph{all} ideal representations accepted by an NFA at the same time is not obvious: For example, we cannot just merge two transitions $p\xrightarrow{\Atomt{\{a\}}}q\xrightarrow{\Atomi{a}}r$, since each of them might be needed for other paths. We achieve this using transducers. 

\vspace{-2mm}
\subparagraph{Transducers} A transducer is a tuple $\tr = \langle \trQ, \trAlph^i, \trAlph^o, \trinit, \trF, \trE \rangle $ where $\trQ$ is a finite set of states, 
$\trAlph^i$ and $\trAlph^o$ are finite (input and output) alphabets, $\trinit \in \trQ$ is the initial state, $\trF \subseteq \trQ$ is the set of final states and $E \subseteq \trQ \times (\trAlph^i \cup \{\emp\})\times (\trAlph^o \cup \{\emp\}) \times \trQ$ is the transition relation.
Each transition, reads a letter (or $\emp$) from the input alphabet, writes a letter (or $\emp$) from the output alphabet and moves to a new state.
A sequence $r = (q_1, a_1, b_1, q_2) (q_2, a_2, b_2, q_3) \cdots (q_{m}, a_{m}, b_m, q_{m+1} )$ is called a run of $\tr$ if each $(q_i, a_i, b_i, q_{i+1})$ is in $\trE$ for $i \in [1, m]$ with $\trinit = q_1$ and $q_{m+1} \in \trF$.
For such a run, let the projection of the transitions to the input (similarly, output) alphabet be denoted by $\inp(r)$ (similarly, $\out(r)$). That is, for the $r$ defined above 
$\inp(r)$ is the subword of $a_1 a_2 \ldots a_m$ and $\out(r)$ is the subword of $b_1 b_2 \ldots b_m$.
Then, the set of $(\inp(r), \out(r))$ over runs $r$ of $\tr$ is called the language of $\tr$, is denoted by $L(\tr)$, and defines a rational relation on $\trAlph^i \times \trAlph^o$. For a set $Y \subseteq (\trAlph^i)^*$, $\tr(Y)$ denotes the set of words $\tr$ outputs upon a run from $Y$, i.e. $\tr(Y)= \{b \mid a \in Y \text{ and } (a, b) \in L(\tr)\}$.

Composition of two transducers is again a transducer~\cite{Shallitbook2008,Berstel1979}. 

\vspace{-2mm}
\subparagraph{Left- and right-reduced representations}
We will later define two transducers $\trl$ and $\trr$ the composition of which will take an ideal representation $\at_1 \cdots \at_n$ produced by $\NFAidl$ and return an equivalent reduced ideal representation $\atb_1 \cdots \atb_m$ with $m \leq n$.
In particular, $\trl$ will turn $\at_1 \cdots \at_n$ into an equivalent ideal representation that is \emph{left-reduced}, and $\trr$ will turn it into an equivalent ideal representation that is \emph{right-reduced}. \emph{Left-} and \emph{right-reducedness} are defined as follows.
An ideal representation $\at_1\cdots \at_n$ is called \emph{left-reduced} if for all $i \in [1, n-1]$, $\at_i$ does not absorb $\at_{i+1}$. 
Similarly, it is called \emph{right-reduced} if $\at_{i+1}$ does not absorb $\at_{i}$. 
Clearly, if a representation is both left- and right-reduced, then it is reduced.

\vspace{-2mm}
\subparagraph{Building the transducers}
Intuitively, the transducer $\trl$ scans an ideal representation and after reading its first alphabet atom $\Atomt{\Delta}$, it outputs $\Atomt{\Delta}$, but then skips (i.e.\ reads without producing output) all atoms that are absorbed by $\Atomt{\Delta}$.
Formally, we have $\trl = \langle \ltrQ, \ltrAlph^i, \ltrAlph^o, \ltrinit, \ltrF, \ltrE \rangle$. $\ltrQ$ contains a state corresponding to each alphabet atom in $\Gamma\subseteq\atoms{\Sigma}$ (as given in~\cref{lemma:idealNFA}), with a new initial state $\ltrinit$ and a state $t_1$ for all of the single atoms.
We set $\ltrAlph^i =  \ltrAlph^o = \Gamma$. Let $\Phi$ be a mapping from $\Gamma$ to $\ltrQ$, which sends each alphabet atom to its corresponding state, and each single atom to $t_1$.
$\ltrF = \ltrQ \setminus \{\ltrinit\}$ and $\ltrE$ is defined as follows;
\begin{itemize}
    \item For all $\atom \in \Gamma$, $(\ltrinit, \atom, \atom, \Phi(\atom))$ and $(t_1, \atom, \atom, \Phi(\atom))$ are in $\ltrE$,
    \item For each $t \in \ltrQ \setminus \{\ltrinit, t_1\}$ and each atom $\atom \in \Gamma$, if $t$ (as an alphabet atom) does not absorb $\alpha$ (as an atom),
         then $(t, \atom, \atom, \Phi(\atom))$ is in $\ltrE$. 
    \item For each $t \in \ltrQ \setminus \{\ltrinit, t_1\}$ and each atom $\atom \in \Gamma$,  if $t$ absorbs $\alpha$,
        then $(t, \atom, \emp, t)$ is in $\ltrE$. 
\end{itemize}

It is easy to see that for an ideal representation $\at_1 \cdots \at_n$, $\trl(\at_1 \cdots \at_n)$ is left-reduced and represents the same ideal.
Clearly the size of $\ltrQ$, as well as the sizes of the alphabets is polynomial. 
Furthermore, for a word $w$ and all $(\inp(w), \out(w))$, $|\out(w)| \leq |\inp(w)|$.

Reversing the edges and flipping the initial and final states of $\trl$ we obtain the reverse transducer $\trr$.
It can be inductively shown that for any ideal representation $\at_1 \cdots \at_n$, $\trr(\at_1 \cdots \at_n)$ is right-reduced. To show \cref{lemma:reducedidealNFA}, we will apply the composition $\trl\circ\trr$ to $L(\NFAidl)$. Here, we need to show that applying $\trl$ after $\trr$ does not spoil right-reducedness:

\begin{lemma}\label{lem:right-normality}
If $\at_1 \cdots \at_n$ is right-reduced, then $\trl(\at_1 \cdots \at_n)$ is also right-reduced.
\end{lemma}
\begin{proof}
    Since $\at_1 \cdots \at_n$ is right-reduced, for $i \in [1,n]$, $\at_{i+1}$ does not absorb $\at_i$.
    By construction, $\ltr(\at_1 \cdots \at_n)$ is a subword of $\at_1 \cdots \at_n$, say $\at_{i_1} \cdots \at_{i_k}$. 
    We show that for all $j \in [1, k]$, $\at_{i_{j+1}}$ does not absorb $\at_{i_j}$. 
    Let $i_j = k$ and $i_{j+1} = k'$. If $k' = k+1$, then the claim follows from the right-reducedness of $\at_1 \ldots \at_n$.
    Otherwise, $\at_{k}$ absorbs all atoms between itself and $\at_{k'}$. In particular it absorbs $\at_{k'-1}$. Since $\at_{k'}$ does not absorb $\at_{k'-1}$, 
     it cannot absorb $\at_{k}$. 
\end{proof}

Thus, by applying $\trl\circ\trr$ to $L(\NFAidl)$, we obtain an NFA that reads ideal representations of the same set of ideals (i.e.~$\dcl{L(\NFA)}$) and every ideal representation is reduced. The resulting NFA $\NFAred$ can be computed in $\NL$. For details of the construction, see~\cref{corollary:complexitycompositiontransducer}.

\subsection{Deciding directedness}
We now present our algorithms to decide directedness for a given NFA. We first
complete the proof of \cref{nfa-ideal-candidate}.  For this, in light of
\cref{lemma:reducedidealNFA} and \cref{prop:weightfunc}, it remains to compute a
path of $\NFAred$ of maximal weight.

\vspace{-2mm}
\subparagraph{Computing the maximal weight.} It is well known that the maximum weight path problem
can be reduced to matrix multiplication over the max-plus semiring~\cite[Lemma~5.11]{bookAHU74}.
This yields $\AC^1$ (resp.\ $\NL$) algorithms for binary (unary) encoded weights~\cite{DBLP:journals/iandc/Cook85}.
For completeness sake we provide a proof of this fact.
Let $\NFAred = (\Qred, \Gamma, \deltared, \qinitred, \Fred)$ be the NFA from
\cref{lemma:reducedidealNFA}.  We may assume $\NFAred$ has a unique final state
$\qfinalred$, is acyclic except for an $\emp$ self-loop in $\qfinalred$, and between each pair of states, there is at most one edge. 
The goal is to compute, for each state $q$, the maximal weight of any path starting in $q$.

Let $\maxlength = |\Qred|$. Since $\NFAred$ is acyclic apart from the self-loop on $\qfinalred$, any ideal representation accepted by $\NFAred$ has length $\leq \maxlength$. 
Observe that due to the $\emp$-loop on $\qfinalred$, every ideal $\at_1 \cdots \at_n \in L(\NFAred)$ is read on some path of length \emph{exactly} $\maxlength$. 
We want to find an accepted path with the maximum summation of $\maxlength$-weights of each transition (recall that $\mu_\maxlength(\emp) = 0$).
To do so, we fix an order $\{q_1, \ldots, q_\maxlength\}$ on the states of $\Qred$ such that $q_1 := \qinitred$ and $q_\maxlength:= \qfinalred$ and construct a $\maxlength\times \maxlength$-matrix $\mymatrix$, the elements of which takes values from the \textit{max-plus semiring} $(\mathbb{N}\cup \{-\infty\},+, \max, 0, -\infty)$.
For each $i,j\in[1,\maxlength]$, we set $\mymatrix(i,j)$ to
 (i) $\mu_\maxlength(x)$, if there is an edge $(q_i, x, q_j) \in \deltared$ with $x\in\atoms{\Sigma} \cup\{\emp\}$,
  (ii) $-\infty$, otherwise.
We can now apply the standard fact from weighted automata that for every $n\ge
0$, in the matrix power $\mymatrix^n$, the entry $(i,j)$ is the maximum weight
of all paths of length exactly $n$ from $q_i$ to $q_j$~\cite{bookAHU74}.
Therefore, the largest weight among all paths from $q_s$ to $q_\maxlength$ is
the entry $(s,\maxlength)$ in the matrix power $\mymatrix^\maxlength$.  A
single matrix product can be computed in $\AC^0$ since binary addition and the
maximum of multiple numbers can be computed in $\AC^0$.  Moreover, for $n$
given in unary, a matrix power $\mymatrix^n$ can be computed in $\AC^1$ by
repeated squaring: One writes $n=\sum_{i=0}^\ell b_i2^i$ with
$b_0,\ldots,b_\ell\in\{0,1\}$ and computes $\mymatrix'_{0}=\mymatrix^{b_\ell}$,
$\mymatrix'_{i}=(\mymatrix'_{i-1})^2\cdot\mymatrix^{b_{\ell-i}}$, yielding
$\mymatrix^n=\mymatrix'_\ell$.  Thus, by applying the (constant depth) $\AC^0$
circuit for matrix multiplication $\ell\in O(\log n)$ times, we obtain a
circuit of logarithmic depth.  In particular, we can compute
$\mymatrix^\maxlength$, and hence the maximal path weights
$\mymatrix^\maxlength(s,\maxlength)$, in $\AC^1$.

In case the alphabet is fixed, we compute the maximal weights in $\NL$. Observe
that in this case, all weights $\mu_\maxlength(\at)$ for atoms
$\at\in\atoms{\Sigma}$  have $\mu_\maxlength(\at)\le (\maxlength+1)^{|\Sigma|}$,
which is polynomial as $|\Sigma|$ is constant.  In particular, all the maximal
path weights from $q_s$ to $q_\maxlength$, are
bounded by $\maxlength\cdot(\maxlength+1)^{|\Sigma|}$ and can be stored in
logarithmic space.  Thus we can proceed as follows.  Given $s\in[1,m]$,
for every $\ell=\maxlength\cdot(\maxlength+1)^{|\Sigma|},\ldots,0$, we
decide in $\NL$ whether there exists a path of weight $\ell$ from $q_s$ to
$q_\maxlength$. If so, then $\ell$ is the maximal weight. Since $\NL=\coNL$, we
can also determine the non-existence of such a path and continue with $\ell-1$.

\vspace{-2mm}

\begin{figure}
    \centering
    \begin{tikzpicture}
        \node (image) at (0,0) {\includegraphics[width=\linewidth]{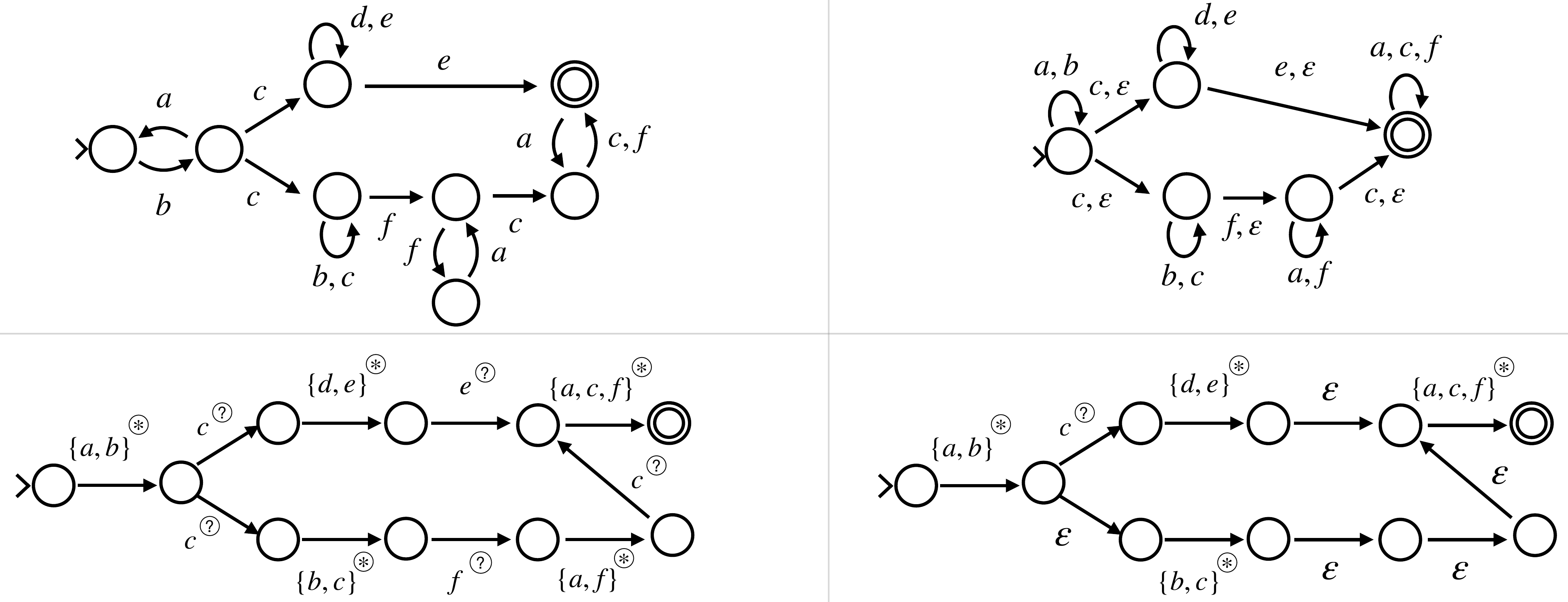}};
        \node at (-6.8,2.2) {\Large{$\NFA$}};
        \node at (1,2.2) {\Large{$\NFAscc$}}; 
        \node at (-6.75,-0.73) {\Large{$\NFAidl$}}; 
        \node at (1.1,-0.73) {\Large{$\NFAred$}}; 
    \end{tikzpicture}
    \caption{Initial automaton $\NFA$ on alphabet $\{a,b,c,d,e,f\}$ and the corresponding $\NFAscc$, $\NFAidl$ and $\NFAred$ are depicted. The initial states of the automata are marked with a half arrow sign and the final states are encircled.
    Since $\NFAred$ has $10$ states all ideal representations in $L(\NFAred)$ contain $\leq 10$ atoms. Therefore, $m$ is set to $10$ and we calculate the maximal $10$-weight ideal, which is 
    $I = \Idl(\{a, b\}^{\circstr} \cdot c^{\circqm} \cdot\{d, e\}^{\circstr} \cdot\{a,c,f\}^{\circstr})$ with the $10$-weight ${11}^3 + 2 \cdot {11}^2 +1$. $L(\NFA) \downarrow \not \subseteq I$ witnessed by $cb \in L(\NFA) \downarrow \setminus I$; proving that $L(\NFA)$ is not directed.}
    \label{fig:running-ex}
    \vspace{-4mm}
\end{figure}

\vspace{-1mm}
\subparagraph{Computing a maximal-weight ideal representation.} 
We have now computed, for each $s$, the maximal weight $M_s$ of any path from $q_s$ to the final state $q_\maxlength$.
For \cref{nfa-ideal-candidate}, we now need to compute in $\NL$ a path from $q_1$ to $q_\maxlength$ of maximal weight.
Here, it is important that this computation only depends on the input (even though our $\NL$ computation is non-deterministic).
Starting from $q_1$, we successively compute the next transition in our path. If $q_i$ is the current state,
then we compute the next state $q_j$ as follows.
We compute $M$ as the maximal $M_\ell$, where $q_\ell$ ranges over all states reachable in one step from $q_i$. Then, we pick the smallest $j$ with $M_j=M$. This way, we successively output a path of maximal weight, such that the path only depends on the input.
This completes \cref{nfa-ideal-candidate}.
%
%

\vspace{-3mm}
\subparagraph{Deciding directedness}
The upper bounds in \cref{thm:main-regularAC1,thm:main-regularNL-complete}
follow by applying \cref{nfa-ideal-candidate} to obtain an ideal
$\at_1\cdots \at_n$, either in $\AC^1$ if $\Sigma$ is part of the
input, or in $\NL$ for fixed $\Sigma$.  Finally, checking whether 
$\dcl{L(\NFA)}\subseteq \Idl(\at_1\cdots \at_n)$ can be done in
$\NL$~\cite{icalpZetzsche16}.




\vspace{-2mm}

\section{Solution on Context-free Languages}\label{sec:solutionCFL}
We now prove \cref{thm:main-contextfree}. As in \cref{sec:solutionregular}, the upper bound uses the weighting
function to compute a candidate ideal in $\dcl{L(G)}$. However, the ideal
representation may be exponentially long and will thus be compressed by a
straight-line program.  For the lower bound, the key idea is to
employ a construction from~\cite{DBLP:conf/stacs/BaumannGTZ22} to compute a
compressed ideal that contains all words of some length $N$ (given in binary),
except a particular word specified as an SLP.

\vspace{-2mm}

\subsection{Deciding directedness of context-free languages}

We begin with the $\PSPACE$ upper bound, which requires some terminology. A \emph{context-free grammar (CFG)} is a tuple $G = \langle N, \Sigma, \pro, S \rangle$ where $N$ is the finite set of \emph{nonterminals}, $\Sigma$ is the finite set of \emph{terminals}, or the finite \emph{alphabet}, 
$S \in N$ is the \emph{start nonterminal} and $\pro \subseteq N \times (N \cup \Sigma)^*$ is the finite set of \emph{productions}.
We use the arrow notation to denote productions. $A \to w$ denotes $(A, w) \in \pro$.
We write $w \to^* w'$ for some $w, w' \in (N \cup \Sigma)^*$ to express that $w'$ can be produced by $w$ through a finite sequence of productions.

For $w \in (N \cup \Sigma)^*$, we denote by $L(w) = \{ w' \in \Sigma^* \mid w \to^* w' \}$ all sequences of terminals $w$ can produce and call it the \emph{language of $w$}.
We define the language of a grammar to be the language of its start nonterminal. That is, $L(G) := L(S)$. 
WLOG we assume that all nonterminals are reachable from $S$. \ISinside{I am not sure if this is needed.}
A CFG is said to be in \emph{Chomsky Normal Form (CNF)}, if all its productions are of the form $A \to BC$, $A\to a$ or $S \to \emp$, where $A,B,C,S \in N$, $a \in \Sigma$ and $B,C \neq S$. It is well known that one can bring a given grammar into CNF in polynomial time.
A CFG $G$ is called \emph{acyclic}, if non of its nonterminals produce itself.

A \emph{straight line program (SLP)} is a CFG that produces a single word.
Formally, an SLP is a CFG $G = \langle N, \Sigma, \pro, S \rangle $ where 
(i) for each $A\in N$, there is exactly one production $A \to w$ in $\pro$,
(ii) $G$ is acyclic. 
We denote the unique word $a_1 \cdots a_n$ an SLP $\mathbb{A}$ produces by $\val(\mathbb{A})$.
If the letters of $\mathbb{A}$ belong to $\atoms{\Sigma}$, $\val(\mathbb{A})$ is an ideal representation over $\Sigma$. 
Thus $\mathbb{A}$ is a 
\emph{compressed ideal representation for} $I = \Idl(\val(\mathbb{A}))$, or shortly \emph{compressed ideal $I$}. 

Our algorithm is analogous to the one for NFAs. First, an analogue of \cref{nfa-ideal-candidate}:
\begin{lemma}\label{cfl-ideal-candidate}
    There is a polynomial time algorithm that given a non-empty CFG $G$, computes a compressed ideal $I \subseteq \dcl{L(G)}$ such that $\dcl{L(G)}$ is directed if and only if $\dcl{L(G)} \subseteq I$.
\end{lemma}
And given \cref{cfl-ideal-candidate}, it remains to decide whether $L(G)\subseteq I$:
\begin{lemma} \label{inclusion-cfl-ideal}
    Given any CFG $G$ and a compressed ideal $I$, one can decide in $\PSPACE$ whether $L(G) \subseteq I$.
\end{lemma}

\vspace{-0.5cm}
\subparagraph{A grammar of ideals} 
The remainder of this subsection is devoted to \cref{cfl-ideal-candidate,inclusion-cfl-ideal}.
Analogously to \cref{lemma:idealNFA}, we
first transform $G$ into an acyclic grammar $\Gidl$ that produces ideal
representations of an ideal decomposition of $\dcl{L(G)}$.
\begin{restatable}{lemma}{restatableidealG}\label{lem:idealG}
Given any CFG $G$ in CNF over $\Sigma$, one can compute in polynomial time an acyclic CFG $\Gidl$ over a polynomial-sized alphabet $\Gamma\subseteq\atoms{\Sigma}$ with
$\Idl(L(\Gidl))=\dcl{L(G)}$.
\end{restatable}

The procedure is similar to Courcelle's construction~\cite{Courcelle91} (see~\cref{app:sec:proofsofsecsolutionCFL} for details).

\vspace{-2mm}
\subparagraph{Reducing ideals} The next step is analogous to \cref{lemma:reducedidealNFA}: We want to transform $\Gidl$ so as to only produce reduced ideal representations. Luckily, we can directly apply the transducers $\trl$ and $\trr$ constructed for \cref{lemma:reducedidealNFA}: Since for a given CFL $K$ and a transducer $\tr$, one can compute in polynomial time a grammar for $\tr(K)$~\cite{Shallitbook2008}, we obtain the following:
\begin{lemma}\label{lemma:reducedidealG}
	Given any CFG $G$ in CNF over alphabet $\Sigma$, one can compute in polynomial
time an acyclic CFG $\Gred$ in CNF over some polynomial-sized alphabet
 $\Gamma\subseteq\atoms{\Sigma}$ such that (i)~$\Idl(L(\Gred))=\dcl{L(G)}$ and
(ii)~all ideal representations in $L(\Gred)$ are reduced.
\end{lemma}
Similar to \cref{lemma:reducedidealNFA}, we apply $\trl \circ \trr$ to
$L(G^\idl)$ (see~\cref{lem:shallitCFL}) and convert to CNF.

\vspace{-2mm}
\subparagraph{Calculating the maximum weight ideal}\label{subsec:CFLcalculatingmaxweightideal}
Similar to \cref{sec:solutionregular}, the next step is to compute for each nonterminal $A$ of $\Gred$ the maximal weight of any ideal representation produced by $A$. 
Let $\Gred = \langle \Nred, \Gamma, \prored, \Sred \rangle$ denote the grammar from \cref{lemma:reducedidealG}. 
With the same argument as in \cref{sec:solutionregular}, an ideal of maximal weight will be as desired in \cref{cfl-ideal-candidate}. We use the weighting function $\mu_m$, where $m=3\cdot 2^{2|\Nred|}$ is an upper bound on the length of words in $\Gred$. A notable difference to \cref{sec:solutionregular} is that here $m$ is exponential. 

For each nonterminal $A$ of $\Gred$, we denote by $\mu_m(A)$ the maximal possible weight of any ideal representation generated by $A$. To calculate $\mu_m(A)$ for each $A$, we employ a simple dynamic programming approach. 
We maintain a table $T$ that contains for each nonterminal $A$ a number $T(A)\in\N$, which is the maximal weight of a derivable ideal representation observed so far. We initialize $T(A)=-\infty$ for every $A$.
Then, we set $T(A)$ to the maximal value of $\mu_m(a)$, where $a$ ranges over all $a\in\atoms{\Sigma}$ for which $A\to a$ is a production. 
Finally, we perform the following update step. For each nonterminal $A$, if there is a production $A\to BC$ such that currently $T(A)$ is smaller than $T(B)+T(C)$, then we update $T(A):=T(B)+T(C)$. It can be shown by induction that after $i$ update steps, $T(A)$ contains the correct value $\mu_m(A)$ for each nonterminal $A$ that has a depth $\le i$ derivation tree that attains $\mu_m(A)$. When we apply the update step $|\Nred|$ times, we arrive at $T(A)=\mu_m(A)$ for every nonterminal $A$.  

\vspace{-2mm}
\subparagraph{Computing the candidate ideal.} Given the numbers $\mu_m(A)$, it
is easy to prove \cref{cfl-ideal-candidate}. For each nonterminal $A$,
there must exist a ``max-weight'' production $A\to BC$, resp. $A\to a$, such
that $\mu_m(A)=\mu_m(B)+\mu_m(C)$, resp.\ $\mu_m(A)=\mu_m(a)$. We build a new
grammar $\SSI$ by selecting for each nonterminal $A$ of $\Gred$ this
max-weight production. Then $\SSI$ contains at most one production for each
nonterminal and is thus an SLP. Moreover, it clearly generates an ideal
representation $\val(\SSI)$ of maximal weight.  
We only need to argue that $L(G)$ is directed iff
$L(G)\subseteq\Idl(\val(\SSI))$. As before, the ``if'' direction is obvious,
because $L(G)\subseteq\Idl(\val(\SSI))$ implies that $L(G)$ is an ideal.
Conversely, suppose $L(G)$ is directed and let $\dcl{L(G)}=I_1\cup\cdots\cup
I_n$ be the ideal decomposition given by $L(\Gred)$ with
$\Idl(\val(\SSI))=I_i$. Since $L(G)$ is directed, $\dcl{L(G)}$ is an ideal and
thus $I_1\cup\cdots\cup I_n=I_j$ for some $j$ by \cref{ideals-basic-facts}. In
particular, we have $I_i\subseteq I_j$. Moreover, if the inclusion were strict,
$I_i$ would not have maximal weight. Hence, $I_i=I_j$ and thus $L(G)\subseteq
I_i=\Idl(\val(\SSI))$ as required.

\vspace{-2mm}
\subparagraph{Deciding directedness}\label{subsec:PSPACEub}
With \cref{cfl-ideal-candidate} in hand, it remains to prove
\cref{inclusion-cfl-ideal}.  Suppose we are given a grammar $G$ and an SLP
$\SSI$ for $I=\Idl(\val(\SSI))$, and we want to check $L(G)\subseteq\val(\SSI)$.  Since this is
equivalent to $\dcl{L(G)}\subseteq\val(\SSI)$, we first construct $\Gred$ given
in~\cref{lemma:reducedidealG}.  Recall that $L(\Gred)$ generates
representations of ideals of $\dcl{L(G)}$. The algorithm guesses an
ideal representation in $L(\Gred)$ whose ideal \emph{does not embed in $I$}.

We guess an ideal representation generated by $\Gred$, atom by atom, via its leftmost derivation. 
This word can be exponentially long, but we only store one (polynomial-length) path in the derivation tree, leading to the terminal atom that we are currently guessing (see~\cref{app:sec:PSPACEub} for an example).
While guessing the representation, we simultaneously maintain a (binary encoded) pointer into $\val(\mathbb{S})$. 
Suppose $\alpha_1 \cdots \alpha_{j-1}$ is guessed so far.
While $\alpha_j$ is being guessed, the pointer holds the length of the shortest prefix of $\val(\mathbb{S})$, $\alpha_1 \cdots \alpha_{j-1}$ embeds in. 
Let $\val(\mathbb{S})[i]$ denote the $i^{th}$ index of $\val(\mathbb{S})$.
If there is an atom $\val(\mathbb{S})[i']$ with $i'\geq i$ (if $\val(\mathbb{S})[i]$ is an alphabet atom) or $i' > i$ (if $\val(\mathbb{S})[i]$ is a single atom) that $\alpha_j$ embeds in, we update the pointer to the smallest such $i'$. If there is no such atom, the guessed ideal does not embed in $I$. On the other hand, if $j-1$ is the last atom guessed, then the guessed ideal embeds in $I$.
Details are in~\cref{app:sec:PSPACEub}.
This establishes \cref{inclusion-cfl-ideal} and thus \cref{thm:main-contextfree}.

\subsection{$\PSPACE$ Lower Bound}\label{subsec:PSPACElb}
Let us now come to the lower bound in \cref{thm:main-contextfree}. It remains to show:
\begin{restatable}{lemma}{restatableCFLthm}\label{thm:cfl-directed-hardness}
	Given a CFG $G$ over $\{0,1\}$, directedness of $L(G)$ is $\PSPACE$-hard.
\end{restatable}
To this end, we reduce from compressed membership in automatic relations.
Given two words $u = a_1 \cdots a_n, v = b_1 \cdots b_n \in \{0,1\}^*$,
their \emph{convolution} is defined as $u \otimes v = (a_1,b_1) \cdots (a_n,b_n) \in (\{0,1\} \times \{0,1\})^*$. The following was shown in \cite[Corollary~8]{DBLP:journals/iandc/Lohrey11}:
\begin{lemma}\label{lemma:LohreyPSPACEhardness}
    There exists a regular language $R \subseteq (\{0,1\} \times \{0,1\})^*$ such that for given two SLPs $\mathbb{A}$ and $\mathbb{B}$
    with $|\val(\mathbb{A})| = |\val(\mathbb{B})|$, deciding 
   $\val(\mathbb{A}) \otimes \val(\mathbb{B}) \in R$  is $\PSPACE$-hard.
\end{lemma}

From \cref{lemma:LohreyPSPACEhardness}, we deduce the following:
\begin{lemma}\label{lemma:membershipPSPACEhardness}
 Given an SLP $\mathbb{B}$ and a CFG $G$ such that all words in $L(G)$ have length exactly $|\val(\mathbb{B})|$, both over the alphabet $\Sigma = \{0,1\}$, deciding $\val(\mathbb{B}) \in L(G)$ is $\PSPACE$-hard.
\end{lemma}
\begin{proof}
We reduce from the $\PSPACE$-complete problem in~\cref{lemma:LohreyPSPACEhardness}.
    Let $R$ be the regular language from~\cref{lemma:LohreyPSPACEhardness}, and let $\mathbb{A}$ and $\mathbb{B}$ be SLPs
    with $n = |\val(\mathbb{A})| = |\val(\mathbb{B})|$.
    Observe that $\val(\mathbb{A}) \otimes \val(\mathbb{B}) \in R$ if and only if $\val(\mathbb{B})$ belongs to the language
    $K=\{ w \in \{0,1\}^n \mid \val(\mathbb{A}) \otimes w \in R \}$.
    We can construct a context-free grammar $G$ for $K$ in polynomial-time
    by viewing an automaton for $R$ as a transducer and applying it to the SLP $\mathbb{A}$.
    \end{proof}

Now in order to reduce the compressed membership problem $\val(\mathbb{B})\in
L(G)$ in \cref{lemma:membershipPSPACEhardness} to an inclusion $L(G)\subseteq
I$, the key trick is to construct an ideal $I$ that acts like a complement of
$\{\val(\mathbb{B})\}$. We expect that this will be of independent interest. 
\begin{lemma}\label{complement-ideal}
	Given an SLP $\mathbb{B}$ over $\Sigma$, one can construct in polynomial time an SLP 
	$\mathbb{I}$ over $\atoms{\Sigma}$ so that $\Idl(\val(\mathbb{B}))$ is infinite and $\Idl(\val(\mathbb{I}))\cap \Sigma^{|\val(\mathbb{B})|}=\Sigma^{|\val(\mathbb{B})|}\setminus\{\val(\mathbb{B})\}$.
\end{lemma}
The proof uses a construction from \cite{DBLP:conf/stacs/BaumannGTZ22}.  The
authors of the latter were interested in defining languages in the existential
fragment of first-order logic over the structure the set $\Sigma^*$, ordered by
$\sw$. In one step \cite[Lemma~3.1]{DBLP:conf/stacs/BaumannGTZ22}, given a word
$u\in\Sigma^*$, they construct a word $\bar{w}\in\Sigma^*$ such that
$\dcl{\{\bar{w}\}}\cap\Sigma^{|w|}=\Sigma^{|w|}\setminus\{w\}$. To this end,
they write $w=a_1\cdots a_n$ and define $u_i$ to be a word that contains every
letter from $\Sigma$, except for $a_i$. Then, they argue that
$\bar{w}=u_1a_1\cdots u_{n-1}a_{n-1}u_n$ is as desired. Here, we cannot use
$\bar{w}$ directly, because we want $I$ to be infinite. However, we can use a
similar construction. 
\begin{proof}[Proof of \cref{complement-ideal}]
    Suppose $\val(\mathbb{B}) = b_1 \cdots b_n$ and define $I=\Idl(w)$ with
    \[ w=\Atomt{(\Sigma \setminus \{b_1\})}\cdot \Atomi{b_1}\cdot \Atomt{(\Sigma \setminus \{b_2\})} \cdot \Atomi{b_2} \cdots \Atomt{(\Sigma \setminus \{b_{n-1}\})}\cdot \Atomi{b_{n-1}}\cdot\Atomt{(\Sigma \setminus \{b_n\})}. \]
    We first show $I\cap\Sigma^{n}=\Sigma^n\setminus\{\val(\mathbb{B})\}$. Clearly, $b_1 \cdots b_n \not \in I$: Let $j_i$ be length of the shortest prefix of $w$ whose ideal contains $b_1 \cdots b_i$. 
    Clearly $j_1=2$, since the first atom $\Atomt{(\Sigma \setminus \{b_1\})}$ does not embed $b_1$. Inductively we get $j_2 = 4, \ldots, j_{n-1} = 2n-2$, which leaves only the last atom $\Atomt{(\Sigma \setminus \{b_n\})}$ to embed $b_n$, but this is not possible. 
    We now prove $\Sigma^n\setminus\{b_1\cdots b_n\}\subseteq I$. Let $c_1 \cdots c_n \neq b_1 \cdots b_n$ and choose $d$ minimally with $c_d\ne b_d$. 
    Let $h_i$ to be the length of the of shortest prefix of $w$ whose ideal contains $c_1 \cdots c_i$. 
    Then $h_{d-1} = j_{d-1} = 2d-2$. Since $c_d$ differs from $b_d$, it embeds in the $(2d-1)$-th atom $\Atomt{(\Sigma \setminus \{b_d\})}$, i.e. $h_d = 2d-1$.
    The remaining $(n-d)$-length suffix $c_{d+1} \cdots c_n$ embeds in the $2(n-i)$ length suffix $\Atomt{(\Sigma \setminus \{b_d\})} \cdot \Atomi{b_d} \cdots \Atomt{(\Sigma \setminus \{b_{n-1}\})}\cdot \Atomi{b_{n-1}}$:
    Indeed, since $\Atomt{(\Sigma \setminus \{b_k\})}\cdot\Atomi{b_k}$ embeds every letter, the aforementioned suffix even embeds every word from $\Sigma^{n-d}$.

   It remains to be shown that we can compute an SLP for $w$. Note that $w$ is \emph{almost} a homomorphic image of $\val(\mathbb{B})$.
    Given SLP $\mathbb{B}$, we obtain an SLP $\mathbb{I}'$ for $w$ by 
    replacing each production $A\to b$ with $A\to \Atomt{(\Sigma\setminus\{b\})}\Atomi{b}$.
    Then, $\val(\mathbb{I}')=\Atomt{(\Sigma \setminus \{b_1\})}\cdot \Atomi{b_1} \cdots \Atomt{(\Sigma \setminus \{b_n\})} \cdot \Atomi{b_{n}}$.
    To get $w$ exactly, we construct $\mathbb{I}$ so that $\val(\mathbb{I})$ is $\val(\mathbb{I}')$ without its last letter. It is easy to see that this can be done in polynomial time (it follows, e.g.\ from~\cite[Theorem 7.1]{schleimer08}).
       \end{proof}

       We are now ready to prove \cref{thm:cfl-directed-hardness}. Given a
       grammar $G$ and an SLP $\mathbb{B}$ as in
       \cref{lemma:membershipPSPACEhardness}, we use \cref{complement-ideal} to
       construct an SLP $\mathbb{I}$ with
       $\Idl(\mathbb{I})\cap\Sigma^{|\val(\mathbb{I})|}=\Sigma^{|\val(\mathbb{I})|}\setminus\{\val(\mathbb{B})\}$.
       Observe that now $\val(\mathbb{B})\notin L(G)$ if and only if
       $L(G)\subseteq\Idl(\mathbb{I})$. Moreover, observe that $L(G)$ is
       finite and $\Idl(\mathbb{I})$ is infinite. Therefore, the following
       lemma implies that $\val(\mathbb{B})\notin L(G)$ if and only if
       the context-free language $L(G)\cup\Idl(\val(\mathbb{I}))$ is directed, yielding
       $\PSPACE$-hardness.

\begin{lemma}\label{lemma:equivalentproblemdirectedness}
    For finite $L\subseteq\Sigma^*$ and an infinite ideal $I$, we have $L \subseteq I $ iff $L \cup I $ is directed.
\end{lemma}
\begin{proof}
	Clearly, if $L\subseteq I$, then $L\cup I=I$ is an ideal and thus directed. Conversely, suppose $L\cup I$ is directed. Consider an ideal decomposition $\dcl{L} = I_1 \cup \ldots \cup I_n$. Here, all $I_i$ are finite since $L$ is finite. Since $L \cup I$ is directed, the downward closure $\dcl{(L \cup I)}$ must coincide with one of the ideals in the ideal decomposition $\dcl{(L\cup I)}=I_1 \cup \ldots \cup I_n\cup I$.
    Since $I$ is the only infinite ideal, this is only possible with $\dcl{(L \cup I)} = I$. In particular, $L \subseteq I$.
\end{proof}

\section{Downward closure comparison}
\subparagraph{Regular languages} We now show how to obtain
\cref{directed-nfa-dce} as a byproduct of our results.  For the upper bounds,
we use \cref{nfa-ideal-candidate} to compute, in $\AC^1$ resp.~$\NL$, a
candidate ideal $I_i$ for each input language $L_i$. Since $L_1$ and $L_2$ are
directed, we must have $\dcl{L_i}=I_i$ and we can decide in deterministic
logspace whether $I_1=I_2$~\cite{icalpZetzsche16}.  This yields an $\NL$ upper
bound for fixed alphabets and an $\AC^1$ upper bound for arbitrary alphabets.

For the $\NL$ lower bound, we reduce from emptiness of NFAs: Given an NFA
$\NFA$, we may assume that all transitions are labeled with the empty word
$\varepsilon$.  We take an NFA $\NFA'$ that just accepts $\{\varepsilon\}$. Then
$L(\NFA)\ne\emptyset$ if and only if $\dcl{L(\NFA')}=\dcl{L(\NFA)}$, proving
$\NL$-hardness.

\vspace{-2mm}
\subparagraph{Context-free languages} We now show~\cref{directed-cfl-dce}. We first use~\cref{cfl-ideal-candidate} to compute an SLP $\mathbb{A}_i$ for
a candidate ideal for each $L_i$. By directedness, $\dcl{L_i}=\Idl(\val(\mathbb{A}_i))$. 
To decide $\Idl(\val(\mathbb{A}_1))=\Idl(\val(\mathbb{A}_2))$, we use the fact
that two reduced ideal representations yield the same ideal if and only if they
are syntactically identical~\cite[Theorem~6.1.12]{thesisHalfon}. To check
$\val(\mathbb{A}_1)=\val(\mathbb{A}_2)$ we may apply the well-known result of
Plandowski~\cite{Plandowski94} that equality of SLPs can be decided in
polynomial time (see also~\cite{Lohrey12}).

For the $\compP$ lower bound, we reduce from emptiness of CFL: Given a
CFG $G$, we may assume that the only word $G$ can produce
is the empty word (otherwise, just replace all occurrences of terminal letters
with the empty word). We also take a grammar $G'$ with
$L(G')=\{\varepsilon\}$. Then clearly, $L(G)\ne\emptyset$ if and only if
$\dcl{L(G)}=\dcl{L(G')}$, yielding $\compP$-hardness.

\section{Conclusion}
We have initiated the investigation of the directedness problem and
determined the exact complexity for context-free languages and for NFAs over fixed alphabets. Over variable alphabets, we show an $\AC^1$
upper bound for NFAs.  Despite serious efforts, we leave the
exact complexity open.  Note that the complexity of directedness
is the same for DFAs and NFAs~\cref{sec:appendix-dfas}. Also, the complexity of the maximum weight path
problem 
is not known~\cite{DBLP:journals/iandc/Cook85}.

The developed techniques could be of independent interest.  
The idea to analyze ideals by their weights might apply to other procedures for
reachability involving
ideals~\cite{DBLP:conf/stacs/FinkelG09,DBLP:conf/icalp/FinkelG09,DBLP:conf/fsttcs/BlondinFG17,DBLP:journals/iandc/BlondinFM18,DBLP:journals/lmcs/BlondinFM17,DBLP:journals/iandc/LazicS21,DBLP:conf/lics/LerouxS15,DBLP:conf/stacs/LerouxS16}.
Furthermore, our $\PSPACE$ lower bound can be viewed as progress towards
resolving the complexity of the \emph{compressed subword problem}: Our lower bound
applies in particular to deciding $L\subseteq I$ for context-free $L$ and a
compressed ideal $I$. Compressed subword, on the other hand, is equivalent to
deciding $I\subseteq J$ for compressed ideals $I,J$. As mentioned before,
it is a long-standing open problem to close the gap between the $\PP$ lower
bound and the $\PSPACE$ upper bound~\cite{DBLP:journals/iandc/Lohrey11}
(see~\cite{Lohrey12} for a survey) for compressed subword.

The surprisingly low complexity of downward closure equivalence (DCE) for
directed CFL calls for an investigation of further applications of directed
CFL. As previously stated, safety properties of concurrent programs only depend on
the downward closure of the participating
threads~\cite{DBLP:journals/corr/abs-1111-1011,DBLP:conf/icalp/0001GMTZ23,DBLP:journals/lmcs/MajumdarTZ22}.
It is conceivable that deciding
safety~\cite{DBLP:journals/corr/abs-1111-1011,DBLP:conf/icalp/BaumannMTZ20} or
other notoriously difficult problems such as
refinement~\cite{DBLP:conf/icalp/0001GMTZ23a} are more tractable for directed
threads as well.

\label{beforebibliography}
\newoutputstream{pages}
\openoutputfile{main.pages.ctr}{pages}
\addtostream{pages}{\getpagerefnumber{beforebibliography}}
\closeoutputstream{pages}
\bibliography{references}

\label{afterbibliography}
\newoutputstream{pagestotal}
\openoutputfile{main.pagestotal.ctr}{pagestotal}
\addtostream{pagestotal}{\getpagerefnumber{afterbibliography}}
\closeoutputstream{pagestotal}

%


\newpage 

\appendix
\section{Proofs of Section~\ref{sec:preliminaries}}
\restatableexistsnormalform*
\begin{proof}
    If the representation $\at_1 \cdots \at_n$ is not in reduced form, there exists two consecutive atoms $\at_i$ and $\at_{i+1}$ that are absorptive. This implies we can merge $\at_i$ and $\at_{i+1}$, by letting only the absorbing one to remain. It is easy to see that, this does not change the ideal that is represented while shortening the length of its representation. By induction on the length of the representation, it follows that every ideal has a reduced representation.
\end{proof}

\section{Proofs of Section~\ref{sec:solutionregular}}\label{app:sec:proofsofsecsolutionregular}
\restatabledownwardclosureNFA*
\begin{proof}
    Let $\NFA=(Q, \Sigma, \delta, q_0, F)$.
    By collapsing each strongly connected component of $\NFA$ into one state, we obtain a partially ordered NFA $\NFAscc$ that recognizes $\dcl{L(\NFA)}$.
    
    Formally, let $\{S_1, \ldots ,S_n\}$ be the set of SCCs of $\NFA$, which can be computed in $\NL$. Then 
    \[ \NFAscc = (\hat{Q} = \{s_1, \dots, s_n\},\Sigma, \hat{\delta}, \hat{q_0}, \hat{F} )\]
    where
    \begin{itemize}
        \item $\hat{\delta}(s_i, a) = s_j$ for all $q_i \in S_i$ and $q_j \in S_j$ with $\delta(q_i, a) = q_j$,
        \item $\hat{q_0} = s_i$ for the $i$ for which $q_0 \in S_i$,
        \item $s_i\in \hat{F}$ for all $i$ for which $F\cap S_i \neq \emptyset$.
    \end{itemize}
    
    Observe that $\NFAscc$ is partially ordered with respect to the reachability order on $Q$.
    
    \begin{claim}
        $L(\NFAscc) = \dcl{L(\NFA)}$.
    \end{claim}
    
    \begin{claimproof}
        $\mathbf{(\dcl{L(\NFA)}\subseteq L(\NFAscc))}$ Let $q_0 \xrightarrow{a_1} q_1 \xrightarrow{a_2} \dots q_{n-1} \xrightarrow{a_{n}} q_n$ be a path of $\NFA$ with $q_n \in F$ and $w = a_1 \ldots a_n$. All transitions remain in $\NFAscc$ for some $s_0' \xrightarrow{a_1} s_1' \xrightarrow{a_2} \ldots s_{n-1}' \xrightarrow{a_n} s_n'$ where $q_i \in S_i'$ and $s_n' \in \hat{F}$. Thus, $w \in L(\NFAscc)$.
    
        $\mathbf{(L(\NFAscc)\subseteq \dcl{L(\NFA)})}$. Let $s_0 \xrightarrow{a_1} s_1 \xrightarrow{a_2} \dots s_{n-1} \xrightarrow{a_{n}} s_n$ be a path in $\NFAscc$ for some $s_n \in \hat{F}$ and let $w = a_1\dots a_n$. The existence of the transition $s_i \xrightarrow{a_{i+1}} s_{i+1}$ in $\hat{\delta}$ implies the existence of a transition $q'_i \xrightarrow{a_{i+1}} q'_{i+1}$ for some $q'_i\in S_i$ and $q'_{i+1} \in S_{i+1}$ in $\delta$. Since $S_i$ and $S_{i+1}$ are SCCs, there exists a path $q'_i \xrightarrow{v a_{i+1}v'} q'_{i+1}$ for some $v, v' \in \Sigma^*$. It follows that, there exists a path $q_0 \xrightarrow{v_1 a_1 v_2 \dots a_n v_{n+1}} q$ for some $q\in F$.  This gives $w \in \dcl{L(\NFA)}$.
    \end{claimproof}
\end{proof}

\restatableidealNFA*
\begin{proof}
    Let $\NFAscc = (\hat{Q},\Sigma, \hat{\delta}, \hat{q_0}, \hat{F} )$. To construct $\NFAidl = (\Qidl, \Gamma, \deltaidl, \qinitidl, \Fidl )$, we double the state set to $\Qidl = \hat{Q} \cup \hat{Q}^\cpy$, let the initial state $\qinitidl = \hat{q_0}$ and final states $\Fidl = \hat{F}^\cpy$. Let the incoming transitions of each node $q \in \hat{Q}$ be received by $q \in \Qidl$ and direct outgoing transitions of it to $q^\cpy \in \Qidl$. Turn all non self-loop transitions in $\NFAscc$ to single atom transitions in $\NFAidl$.
    If $q$ has a self-loop, then let $q$ go to $q^\cpy$ in $\NFAidl$ with the alphabet atom that contains all letters on the self loop. If $q$ contains no self loops, let $q \xrightarrow{\emp} q^\cpy$. 
    Formally, the transition function $\deltaidl$ and the finite alphabet $\Gamma \subseteq \atoms{\Sigma}$ are set as follows
    \begin{enumerate}
        \item preserve the $\emp$-transitions i.e. for $(q_1, \emp, q_2) \in \hat{\delta}$, add $({q_1}^\cpy, \emp, q_2) \in \deltaidl$,
        \item for each $q$ in $\hat{Q}$ that has a non-$\emp$ self loop in $\NFAscc$, add $\atomt_q= \{a \in \Sigma \mid (q,a,q) \in \hat{\delta}\}$ to $\Gamma$ and add $(q, \atomt_q, q^\cpy) \in \deltaidl$,
        \item for each $q$ in $\hat{Q}$ that has no non-$\emp$ self loop in $\NFAscc$, add $(q, \emp, q^\cpy) \in \deltaidl$,
        \item for each $(q_1, a, q_2) \in \delta$ where $q_1 \neq q_2$ and $a \in \Sigma$, add $\atomi$ to $\Gamma$ and $({q_1}^\cpy, \atomi, q_2)$ to $\deltaidl$. 
 \end{enumerate}
    
    It is easy to see that the language $\NFAidl$ accepts is the ideal decomposition of $\dcl{L(\NFA)}$ given in equation~\eqref{eq:acc-paths-ideals}. 
\end{proof}

\begin{remark}[Complexity of the construction of $\trl$ --~\cref{subsec:reducingideals}]\label{remark:complexityoftrl}
Recall that the size of $\Gamma$ is polynomial due to~\cref{lemma:idealNFA}. Then $\ltr$ can be constructed in $\logspace$, since to add the outgoing edges for each $t$, we only need to go over each $\atom \in \Gamma$ and check if $t$ contains $\alpha$, which takes at most $\log(|\Sigma|)$ space. 
\end{remark}

\begin{lemma}\label{lem:reversedtransducer}
    For a transducer $\tr$, let $\tr^{rev}$ be the transducer obtained by reversing the edges and flipping the initial and final states of $\tr$. For a word $w = a_1 \ldots a_m$ over the input alphabet of $\tr$, let $w^{rev}$ be the reversed word, i.e. $w^{rev} = a_m\ldots a_1$. 
    Then, $\{ v \mid v^{rev} \in \tr^{rev}(w^{rev}) \} = \tr(w)$.
\end{lemma}
\begin{proof} 
    We will continue with the proof as if $\tr^{rev}$ has multiple initial states, which are the final states of $\tr$. In fact we can easily set a unique initial state by adding a new state $q_{init}$ to $\tr^{rev}$ and adding $\varepsilon-$transitions, $(q_{init}, \varepsilon, \varepsilon, q')$ for each final state $q'$ of $\tr$, to comply with our definition of a transducer.
    
    \noindent \textbf{($\supseteq$).} For any word $v \in \tr(w)$, where $ v = \out(r)$ and $w = \inp(r)$ for a run \\ $r = (q_1, a_1, b_1, q_2)(q_2, a_2, b_2, q_3)\cdots(q_m, a_m, b_m, q_{m+1})$ in $\tr$, where $q_1$ is the initial state. 
    Then $r^{rev} = (q_{m+1}, a_m, b_m, q_m)\cdots (q_3, a_2, b_2, q_2) (q_2, a_1, b_1, q_1)$ is a run of $\tr^{rev}$ where $\inp(r^{rev}) = w^{rev}$ and $\out(r^{rev}) = v^{rev}$. That is, $v^{rev} \in \tr^{rev}(w^{rev})$, which shows $\{v \mid v^{rev} \in \tr^{rev}(w^{rev}) \} \supseteq \tr(w)$.
    \newline \noindent \textbf{($\subseteq$).} This side follows from the observation that $\tr = (\tr^{rev})^{rev}$.
\end{proof}
It follows from~\cref{lem:reversedtransducer} that, for a given ideal representation $\at_1\cdots \at_n$, $\trr(\at_1 \cdots \at_n)$ is right-reduced and from~\cref{lem:right-normality} that $\trr$ preserves left-reducedness.
It immediately follows that the composition transducer $\trl \circ \trr$, takes an ideal representation and prints out an equivalent reduced representation. 
\begin{lemma}[\cite{Shallitbook2008} ]\label{lem:shallitregular}
    Given a transducer $\tr$ and an NFA $\NFA$, $\tr(L(\NFA))$ is a regular language. Furthermore, we can calculate an NFA that accepts the language in $\logspace$.
\end{lemma}

\begin{corollary}\label{corollary:complexitycompositiontransducer}
    Due to~\cref{remark:complexityoftrl},~\cref{lem:reversedtransducer} and~\cref{lem:shallitregular}, an NFA that accepts $\trl\circ \trr (\NFAidl)$ can be constructed in $\logspace$.
\end{corollary}

\restatableembedding*
\begin{proof}
    Fix an order on $\Sigma$ and for each $\Delta \subseteq \Sigma$, let $\worddelta{\Delta}$ be some word that contains each letter in $\Delta$ once, in the increasing order.
    For an atom $\at_i$ define $\wordatom{\at_i}$ as follows,
    \begin{align}
        \wordatom{\at_i} = \begin{cases}
            a, &\text{ if } \, \at_i = \atomi,\\
            {\worddelta{\Delta}}^{m+1}, &\text{ if }\, \at_i = \atomt
        \end{cases}
    \end{align}
    where ${\worddelta{\Delta}}^{m+1} $ is the concatenation of $(m+1)$-many $\worddelta{\Delta}$s.
    \newline
    
We now define a function $f$ that satisfies premises of~\cref{line:one}-\ref{line:three}:

Let $f$ send each $i$ to the $j$ for which $\atb_1\cdots \atb_j$ is the shortest prefix of $\atb_1\cdots \atb_m$ for which $w_{\at_1}\ldots w_{\at_i} \in \Idl(\atb_1 \cdots \atb_j)$.
It is easy to see that for this $f$, the premise of~\cref{line:one} trivially holds. 
Now we show that the premise of~\cref{line:two} holds. First assume $\at_i$ is the single atom $\atomi$ and therefore $w_{\at_i}$ = $a$. Assume $\atb_j$ does not contain $\at_i$, i.e. $a \not \in \Idl(\atb_j)$. Then it is easy to see that there exists a prefix $\atb_1 \cdots \atb_{j'}$ with $j' < j$ the language of which contains $w_1 \ldots w_i$.
Now assume $\at_i$ is an alphabet atom $\atomt$. Then $w_{\at_i} = {\worddelta{\Delta}}^{m+1}$. Assume $f(i-1) = j'$, that is $\atb_1 \cdots \atb_{j'}$ is the shortest prefix $w_{\at_1} \ldots w_{\at_{i-1}}$ is in the language of. Then, $w_{\at_i} \in \Idl(\atb_{j'} \cdots \atb_j)$.
Due to the length of ${\worddelta{\Delta}}^{m+1} \geq |\Delta|\cdot(m+1)$, there exists some (smallest) $k \in [j',j]$ and an infix $w'$ of ${\worddelta{\Delta}}^{m+1}$ of length $|\Delta|+1$ such that $w' \in L(\atb_k)$. The fact that $|w'|\geq 2$ implies that $\atb_k$ is an alphabet atom and the fact that each infix of ${\worddelta{\Delta}}^t$ of length $|\Delta|$ contains each letter in $\Delta$ further implies that the alphabet of $\atb_k$ contains $\Delta$, and therefore that $\atb_k$ contains $\at_i$.
This implies that ${\worddelta{\Delta}}^{m+1}  \in \Idl(\atb_k)$ and thus we have $w_{\at_1}\ldots w_{\at_i} \in \Idl(\atb_1 \cdots \atb_k)$, and since $k$ is the smallest number that satisfies this condition, $j=k$. Therefore in both cases, the premise of~\cref{line:two} holds. 
Finally we show that the premise of~\cref{line:three} holds. Due to items~\ref{line:one} and~\ref{line:two}, ~\cref{line:three} can be violated only if there exists $\at_i = \at_{i+1} = \atb_j = \atomi$ with $f(i) = f(i+1) = j$ for some $a \in \Sigma$.
But if $\atb_1\cdots \atb_j$ is the shortest prefix for which $w_{\at_1} \ldots w_{\at_{i-1}} a \in \Idl(\atb_1 \cdots \atb_j)$, then clearly $w_{\at_1} \ldots w_{\at_{i-1}} a a \not \in \Idl(\atb_1 \cdots \atb_j)$ since $\Idl(\atb_j) = \langatomi$, which shows that~\cref{line:three} holds.
\end{proof}

\section{Proofs of Section~\ref{sec:solutionCFL}}\label{app:sec:proofsofsecsolutionCFL}
\restatableidealG*
\begin{proof}
For $L \subseteq \L(G)$, we denote by $\alp(L)$ the set of terminals occuring in $L$.
Since we are interested in calculating a representation of $\dcl{L(G)}$ at the end, we can assume WLOG that the language generated by each nonterminal $A \in N$ is downward closed.
Upon this assumption, we recite some observations by Courcelle.

We observe that any nonterminal $A$ that produces itself twice, i.e. $A \to^* w A w' A w''$ for some $w,w',w'' \in (N \cup \Sigma)^*$, exactly produces $\alp(A)$. 
It is clear that $L(A) \subseteq (\alp(A))^*$. To observe the other side it is sufficient to observe that for any $x, y \in (\alp(A))^*$, $xy$ and $yx$ are both in $L(A)$.
If a nonterminal $A$ produces itself twice, we write $A \Rightarrow^2 A$.

If a nonterminal $A$ produces a nonterminal $B$, i.e. $A \to^* w B w'$ for some $w, w' \in (N \cup \Sigma)^*$, we write $A \Rightarrow^1 B$. Clearly, $A \Rightarrow^2 A$ implies $A \Rightarrow^1 A$.
In the following, assume no nonterminal produces itself twice. Take $A \Rightarrow^1 B \Rightarrow^1 A$. Clearly $B \Rightarrow^1 B$. For all such $A, B$ we write $A \equiv B $ and create equivalence classes
$A_{\equiv} = \{ B \mid B \equiv A\}$. 

Then we know from~\cite{Courcelle91} that each equivalence class generates the same language and this language is defined as follows. 
$$ L(A_{\equiv}) = (\lA)^* ~\dcl{\dot{A}} ~ (\rA)^* \quad \text{where}$$
$$\lA := \bigcup_{X \in A_{\equiv}} \alp(\leftset(X)), \quad \quad \rA := \bigcup_{X \in A_{\equiv}}\alp(\rightset(X))$$ 
We set $\leftset(X), \rightset(X)$ and $\dot{A}$ as follows,
\begin{itemize}
    \item Let $X \in A_\equiv$ and $X \to w$ be a production where $w \in (N \cup \Sigma)^*$ contains no nonterminals that produce $X$. Then $L(w) \in \dot{A}$.
    \item Let $X, Y \in  A_\equiv$ and $X \to w Y w'$ be a production with $w, w' \in  (N \cup \Sigma)^*$. Clearly, $w, w'$ do not contain any nonterminals that produce $X$, since we assumed $X$ does not produce itself twice.
    Then $w \in \leftset(X)$ and $w' \in \rightset(X)$.
\end{itemize}
Lastly, for nonterminals with $A \not \Rightarrow^1 A$, we have the trivial rule $L(A) = \dcl{\dot{A}}$.

 
Remember that the ideals of $\dcl{L(G)}$ are defined as a finite sequence of alphabet and single atoms. 
 Next we construct a CFG $\Gidl$ that produces the language of the ideal representations of $\dcl{L(G)}$.
 
Let $N' \subseteq N$ denote the set of nonterminals in $G$ that produces themselves once, and not twice. We fix a representative nonterminal from each equivalence class $A_{\equiv}$ in $N'$.
For each $A \in N$, if $A\in N'$, we set $\bar{A}$ to the representative of its equivalence class. Otherwise, we set $\bar{A}$ to $A$ itself. 

\hspace{0.3cm}

Let $G= \langle N, \Sigma, \pro, S \rangle$. We follow the below steps to compute $G^\idl= \langle N^\idl, \Gamma, \pro^\idl, S^\idl \rangle$:
\begin{enumerate}
    \item For each $A \in N$, add $\bar{A}$, and for each $A \in N'$ additionally add a new nonterminal $\twobar{A}$ in $N^\idl$.
    \item For each $A \in N$ with $A \Rightarrow^2 A$, add $\bar{A} \to (\alp(A))^*$ in $\pro^\idl$, and $(\alp(A))^*$ in $\Gamma$.
    \item For each $A \in N'$, add $\bar{A} \to (\lA)^* \twobar{A} (\rA)^*$ in $\pro^\idl$ and add $(\lA)^*$ and $(\rA)^*$ in $\Gamma$.
    \item For each $A \to a$ in $\pro$ for some $a \in \Sigma$, if $A \in N'$ add $\twobar{A} \to \{a, \emp\}$, otherwise add $\bar{A} \to \{a, \emp\}$ in $\pro^\idl$ and add $ \{a, \emp\}$ in $\Gamma$.
    \item For each $A \to BC$ in $\pro$ where neither $B$ nor $C$ produces $A$, 
    if $A \in N'$, add  $\twobar{A} \to \bar{B}\bar{C}$, otherwise add $\bar{A} \to \bar{B}\bar{C}$ in $\pro^\idl$.
    \item Add $\{\emp\}$ in $\Gamma$ and for each $Y \in N^\idl$, add $Y \to \{\emp\}$ in $\pro^\idl$.
    \item Set $S^\idl$ to $\bar{S}$.
\end{enumerate}

\textbf{Complexity of the Construction.}
We split $N$  into subsets of nonterminals that 
(i) produce themselves twice, 
 (ii) produce themselves and each other once, and 
 (iii) don't produce themselves.
 This computation takes only linear time, e.g. using depth-first search. 
Computing $\alp(w)$ for a $w \in (N \cup \Sigma)^*$ similarly takes only linear time. 
In order to compute $\lA$ and $\rA$ for each equivalence class $A_{\equiv}$, we need to iterate through the productions and compute $\alp(w)$ and $\alp(w')$ wherever applicable, which can again be accomplished in linear time.
Clearly, the remaining steps do not take more time either, since we only go through each production once and do a linear work. 

In addition, note that the size of $\Gamma$ is linear in $\max\{|\Sigma|,|\pro|\}$ and the size of $N^\idl$ is at most $2 \cdot |N|$.

\textbf{Correctness of the Construction.} We claim that $L(G^\idl)$ is an ideal decomposition of $\dcl{L(G)}$, i.e. the words accepted by the ideals in $\bigcup L(G^\idl)$ are equal to $\dcl{L(G)}$. 
For this, we only need to show that our contruction mimics Courcelle's. For nonterminals that produces themselves twice, the language they generate is $(\alp(A))^*$ by Courcelle's construction. 
Clearly, those words are exactly the ones contained in the ideal $(\alp(A))^*$, which is the language the nonterminal generates, due to Step 2.
For nonterminals that do not produce themselves, Step 4 and 5 make sure that the production with $\bar{A}$ in $\pro^\idl$ generates ideals that capture exactly $L(A)$.
For terminals that produce themselves only once, we need to show that the production in Step 3 gives $(\lA)^* \dcl{\dot{A}} (\rA)^*$. Since the atoms $(\lA)^*$ and $(\rA)^*$ clearly mimic the rule, we only need to show that $L(\twobar{A}) = \dcl{\dot{A}}$. But this is clear since $\twobar{A}$ has the productions that exactly capture the productions of $A \to w$ in $G$ that do not produce $A$. Thus, $L(w)$ is in $L(\twobar{A})$ for all such $w$. 
Since these are the only productions $\twobar{A}$ mimics, and $L(\twobar{A})$ is downward closed, $L(\twobar{A}) = \dcl{L(\dot{A})}$.

\textbf{Ideal representations of $L(G^\idl)$.} 
Observe that $G^\idl$ is \emph{acyclic}. Therefore, the language it generates is finite and the length of the computation tree that yields each ideal representation is at most $|N^\idl|$. Consequently, the length of each ideal representation is at most $3\cdot2^{2 \cdot |N|}$ (due to the production in Step 3 branching 3 times, and in the other branching at most twice).
\end{proof}

\begin{lemma}[\cite{Shallitbook2008} Theorem 4.1.5]\label{lem:shallitCFL}
    Given a transducer $\tr$ and a context free grammar $G$, $\tr(L(G))$ is a context free language. Furthermore, we can calculate a CFG that accepts the language in polynomial time.
\end{lemma}
The time complexity of~\cref{lem:shallitCFL} follows from the proof of the theorem in the citation, due to morphism, inverse morphism and intersection with a regular language all being polynomial time decidable for context-free grammars.
    
\section{Example for~\cref{subsec:PSPACEub}}\label{app:sec:PSPACEub}
\begin{example}\label{ex:ideal-derivation}
    Consider the grammar that produces ideal representations over the nonterminal alphabet $\{S, A, B, C\} $, terminal alphabet $\{\{c, \emp\}, \{d, \emp\}, \{a,b\}^*, \{a,c\}^*\}$ and the productions
    \begin{align*} & S \to AA \mid BB\\
                   & A \to BB \mid CB \mid BC\\
                   & B \to \{a, b\}^* \mid \{d, \emp\}\\
                   & C \to \{a, c\}^* \mid \{c, \emp\}
    \end{align*}
    Assume the ideal representation being guessed is $\{c, \emp\}\{d, \emp\} \{a,b\}^* \{a,c\}^*$.
    The leftmost derivation of this ideal representation and the memory content at each step of the computation are given below side-by-side. The lefthand side depicts the rules that are taken, and the righthand side depicts the memory content during that step.
    \begin{alignat*}{2}
        & S \to \mathbf{A}A \quad \quad && S \to \mathbf{A}A \\
        & S \to \mathbf{C}BA \quad \quad && S \to \mathbf{C}BA \\
        & S \to \{c, \emp\}\mathbf{B}A \quad \quad && S \to \{c, \emp\}\mathbf{B}A \\
        & S \to \{c, \emp\}\{d, \emp\}\mathbf{A} \quad \quad && S \to \{d, \emp\}\mathbf{A} \\
        & S \to \{c, \emp\}\{d, \emp\}\{a, b\}^*\mathbf{C} \quad \quad && S \to \{a,b\}^*\mathbf{C} \\
        & S \to \{c, \emp\}\{d, \emp\}\{a, b\}^*\{a,c\}^* \quad \quad && S \to \{a,c\}^* \\
    \end{alignat*}
    As illustrated, once we generate an atom, we store it only for one step and then drop it. This allows us to generate an exponentially sized ideal representation using polynomial space. Another way to view the derivation tree and the memory content is given in~\cref{fig:ideal-derivation-tree}. 
    Here, the path that is kept in memory while guessing the atom $\{d, \emp\}$ is shown in red color.
    That is, while guessing each atom, we keep the path from the root that leads to the atom, as well as the unexplored righthand side nonterminals of each rule that was taken on the path. 
\end{example}
Since $\Gred$ is acyclic and thus the length of the derivation tree is polynomial, the derivation procedure explained in~\cref{ex:ideal-derivation} can be completed in polynomial space. 
Thus, we can guess any ideal representation generated by $\Gred$, atom by atom, in polynomial space.

\begin{figure} \centering
\begin{tikzpicture}[thick]
    \node[red] {\textbf{S}}
      child[red] {node[xshift=-6mm] {\textbf{A}}
        child[black] {node {C}
          child {node[yshift=3.5mm] {$\{c, \emp\}$}}}
        child {node {\textbf{B}}
           child {node[yshift=3.5mm] {$\mathbf{\{d, \emp\}}$}}}}
      child[red] {node[xshift=6mm] {\textbf{A}}
        child[black] {node {B}
            child {node[yshift=3.5mm] {$\{a, b\}^*$}}}
        child[black] {node {C}
            child {node[yshift=3.5mm] {$\{a, c\}^*$}}}};
  \end{tikzpicture}\caption{Derivation tree of the ideal in~\cref{ex:ideal-derivation}. The path that is kept in the memory while guessing the atom $\{d, \emp\}$ is given in red.}\label{fig:ideal-derivation-tree}
\end{figure}

While we are guessing an ideal representation generated by $\Gred$ as explained, we simultaneously keep and update a (binary encoded) pointer inside $\val(\mathbb{I})$. 
Assume the guessed ideal representation is $\alpha_1 \alpha_2 \cdots \alpha_m$ where each $\alpha_j$ is an atom.
While $\alpha_j$ is guessed, the pointer keeps the length of the shortest prefix of $\val(\mathbb{I})$, $\alpha_1 \cdots \alpha_{j-1}$ embeds in. 
According to $\alpha_j$, we update the pointer value.

The explicit computation is as follows: let $i$ be the length of the shortest prefix of $\val(\mathbb{I})$ that embeds $\alpha_1\cdots \alpha_{j-1}$. If $j=1$, $i$ is taken to be $0$. Then we set 
$$i'=\begin{cases} i+1 &\text{if $i=0$ or $\val(\mathbb{I})[i]$ is a single atom, }\\
i &\text{ otherwise.}\end{cases}$$ 
Then we check whether $\val(\mathbb{I})[i']$ contains $\alpha_j$. If it does, we keep the pointer value at $i'$ and move on to guessing the next atom. If it does not, we increase $i'$ by one and repeat the process. If $i' = |\val(\mathbb{I})|+1$, we declare the word rejected by $\val(\mathbb{I})$. If the production of the guessed ideal has ended, we declare the word accepted by $\val(\mathbb{I})$. 
~\cref{lemma:Lohrey} allows us to compute $\val(\mathbb{I})[i]$ in polynomial time at each call.

\begin{lemma}[~\cite{Lohrey12}]\label{lemma:Lohrey}
Given an SLP and an index $i$ of the unique word $w$ it generates, one can calculate the $i^{th}$ letter of $w$ in polynomial time. 
\end{lemma}

Therefore, the inclusion $\dcl{L(G)} \subseteq I $ can be checked in $\PSPACE$.

\section{Counting ideals}\label{sec:appendix-counting-ideals}
\restatablecountingideals*
\begin{proof}
	By \#NFA, we denote the problem where we are given an NFA $\NFA$ and a
	number $n$ (in unary) and are asked how many words of length $n$ are
	contained in $L(\NFA)$. This problem is well-known to be
	$\sharpP$-complete (more precisely,
	$\spanL$-complete)~\cite{DBLP:journals/tcs/AlvarezJ93}. We show that
	counting ideals is $\sharpP$-hard by reducing \#NFA to it. Given an NFA
	$\NFA$ over an alphabet $\Sigma$, we construct in $\logspace$ an NFA
	$\NFA'$ such that $L(\NFA')=L(\NFA)\cap \Sigma^n$. Now clearly, the
	decomposition into maximal ideals of $\dcl{L(\NFA')}$ is
	$I_1\cup\cdots\cup I_m$, where each
	$I_i=\Idl(\Atomi{a_1}\cdots\Atomi{a_n})$ for some word $a_1\cdots a_n$
	of length $n$ in $L(\NFA)$. In particular, the number of words of
	length $n$ in $L(\NFA)$ is exactly the number of ideals in the ideal
	decomposition of $L(\NFA')$. This proves $\sharpP$-hardness.

	Let us now show that one can count the number of ideals in the ideal
	decomposition of $\dcl{L(\NFA)}$ in $\sharpP$. Suppose $\Sigma$ is the
	alphabet of $\NFA$. We first apply ideal \cref{lemma:reducedidealNFA}
	to compute an NFA $\NFAred$ over some $\Gamma\subseteq\atoms{\Sigma}$.
	Now our task is to compute the number of paths in $\NFAred$ that read
	\emph{maximal} ideals, i.e. ones that are not strictly included in any
	ideal. Our non-deterministic polynomial-time Turing machine
	guesses a word $w$ accepted by $\NFAred$, and thus an ideal
	$I=\Idl(w)\subseteq \dcl{L(\NFA)}$. We want to count this ideal (meaning:
	accept in this branch of the TM) if $I$ is maximal, i.e.~not strictly
	included in any other ideal of $\NFAred$. Thus, it remains to show that
	we can decide in polynomial time whether $I$ is maximal.
	
	We exploit the fact that all paths contain reduced ideals.  An
	inspection of \cref{lemma:embedding} and the proof of
	\cref{prop:weightfunc} shows that it is easy to construct, in
	polynomial time, a finite-state transducer $\tr_\Gamma$ with input
	alphabet $\Gamma$ such that the pair $(u,v)\in\Gamma^*\times\Gamma^*$
	is accepted by $\tr_\Gamma$ if and only if (i)~$u$ and $v$ are reduced
	ideal representations and (ii)~$\Idl(v)\subsetneq\Idl(u)$. Therefore,
	$I$ is maximal if and only if $w\notin \tr(L(\NFAred))$. Since one can
	construct (in $\logspace$) an NFA for $\tr(L(\NFA))$, this can be checked
	in polynomial time.
\end{proof}
\section{Deterministic finite automata}\label{sec:appendix-dfas}
For the directedness problem for NFAs, we have an $\AC^1$ upper bound and an
$\NL$ lower bound. This raises the question of whether perhaps the exact
complexity can be determined more easily for deterministic finite automata
(DFAs). We sketch briefly that this is not the case: There is a simple
deterministic logspace reduction that takes an NFA $\NFA$ and computes a DFA
$\NFA'$ such that $L(\NFA)$ is directed if and only if $L(\NFA')$ is directed.
Hence, the complexity does not depend on whether the input automaton is
deterministic.

\subparagraph{Eliminating $\varepsilon$-transitions} The reduction proceeds in
two steps. In the first step, we reduce to an NFA without
$\varepsilon$-transitions. This is easily done in an $\NL$ reduction (even
while preserving the exact language, not just directedness), but we are
interested in a deterministic logspace reduction. Given the NFA $\NFA$ over
$\Sigma$, we pick a fresh letter $\#\notin\Sigma$ and replace every
$\varepsilon$-transition $p\xrightarrow{\varepsilon} q$ with
$p\xrightarrow{\#}q$. Moreover, we add a self-loop $q\xrightarrow{\#}q$ to
every state. For resulting NFA $\NFA'$, observe that $\dcl{L(\NFA')}$ is the
set of all words in $(\Sigma\cup\{\#\})^*$ such that erasing all $\#$ yields a
word in $\dcl{L(\NFA)}$. It is easy to see that $\dcl{L(\NFA')}$ is directed if
and only if $\dcl{L(\NFA)}$ is directed. In particular, directedness of the
NFAs' languages itself is not affected.

\subparagraph{Achieving determinism} In our second step, we determinize. By the
first step, we assume that we are given an $\varepsilon$-free NFA $\NFA$.
Suppose $\NFA$ has transitions $t_1,\ldots,t_n$. Moreover, suppose transition
$t_i$ is labeled with $a_i\in\Sigma$. We obtain a new automaton $\NFA'$ from
$\NFA$ by labeling each transition $t_i$ with the regular language
$\{b_1,\ldots,b_n\}^*b_ia_i$, where $b_1,\ldots,b_n\notin\Sigma$ are $n$
pairwise distinct fresh letters. The language accepted by $\NFA'$ is defined in
the obvious way: For each transition $t$, it can read an arbitrary word from
the language that labels $t$. Let $K$ be the language of $\NFA'$ defined in this way.

First, note that for $K$, one can easily construct a DFA in $\logspace$.
Essentially, the last letter in each block from $\{b_1,\ldots,b_n\}^*$ tells it
which transition of $\NFA$ to simulate when reading a letter from $\Sigma$
afterwards.

Finally, it is easy to see that $K$ is directed if and only if $L(\NFA)$ is
directed.  We say that a word $w'\in(\Sigma\cup\{b_1,\ldots,b_n\})^*$ is a
\emph{padding} of $w\in\Sigma^*$ if $w$ is obtained from $w'$ by deleting all
occurrences of $b_1,\ldots,b_n$.  Now if $L(\NFA)$ is not directed, then there
are two words $u,v\in\Sigma^*$ such that $u,v\in L(\NFA)$ and there is no
common superword in $L(\NFA)$. But then $u,v$ must have paddings $u',v'$ in $K$
and clearly, $u'$ and $v'$ cannot have a common superword in $K$ (otherwise, we
could delete all $b_1,\ldots,b_n$ from it to obtain a common superword for
$u,v$). 

Conversely, if $L(\NFA)$ is directed, we consider two words $u',v'\in K$. They
are paddings of words $u,v\in L(\NFA)$ and so there is a common superword
$w\in\Sigma^*$, i.e.~$u\sw w$ and $v\sw w$. But $w$ also has a padding $w'$ in
$K$. Moreover, we can insert more occurrences of $b_1,\ldots,b_n$ into $w'$ to
obtain a common superword $w''$ for $u'$ and $v'$. Thus, $K$ is directed as
well.

\newoutputstream{todos}
\openoutputfile{main.todos.ctr}{todos}
\addtostream{todos}{\arabic{@todonotes@numberoftodonotes}}
\closeoutputstream{todos}
\end{document}